\documentclass[11pt]{article}
\usepackage{graphicx} 
\usepackage{amsfonts}
\usepackage{amsmath}
\usepackage{amsthm}
\usepackage{geometry}
\usepackage{subcaption} 
\usepackage[pagebackref]{hyperref}
\hypersetup{
    colorlinks=true,
    linkcolor=blue,
    filecolor=magenta,      
    urlcolor=cyan
}

\usepackage{authblk}
\usepackage{physics}
\usepackage{mathtools}
\usepackage{subcaption}
\usepackage{float}

\theoremstyle{plain}
\newtheorem{corollary}{Corollary}
\newtheorem{proposition}{Proposition}
\newtheorem{theorem}{Theorem}
\newtheorem{lemma}{Lemma}

\theoremstyle{remark}
\newtheorem{example}{Example}
\newtheorem{remark}{Remark}

\newcommand{\Prob}{\mathbb{P}}
\newcommand{\E}{\mathbb{E}}

\newcommand{\cl}[1]{\mathcal{#1}}
\renewcommand{\sf}[1]{\textsf{#1}}
\newcommand{\set}[1]{\left\lbrace #1 \right\rbrace}

\DeclarePairedDelimiter\floor{\lfloor}{\rfloor}

\usepackage{xcolor}

\DeclareMathOperator*{\argmax}{argmax}
\DeclareMathOperator*{\Var}{Var}

\title{Optimizing Returns from Experimentation Programs}

\author[1]{Timothy Sudijono\footnote{Work done while at Netflix. Email: \url{tsudijon@stanford.edu}}}
\author[2]{Simon Ejdemyr}
\author[2]{Apoorva Lal}
\author[2]{Martin Tingley}
\affil[1]{Department of Statistics, Stanford University}
\affil[2]{Netflix}

\date{\today}

\begin{document}
\maketitle
\begin{abstract}
Experimentation in online digital platforms is used to inform decision making. Specifically, the goal of many experiments is to optimize a metric of interest. Null hypothesis statistical testing can be ill-suited to this task, as it is indifferent to the magnitude of effect sizes and opportunity costs. Given access to a pool of related past experiments, we discuss how experimentation practice should change when the goal is optimization. We survey the literature on empirical Bayes analyses of A/B test portfolios, and single out the \textit{A/B Testing Problem} \cite{azevedo2020b} as a starting point, which treats experimentation as a constrained optimization problem. We show that the framework can be solved with dynamic programming and implemented by appropriately tuning $p$-value thresholds. Furthermore, we develop several extensions of the A/B testing problem and discuss the implications of these results on experimentation programs in industry. For example, under no-cost assumptions, firms should be testing many more ideas, reducing test allocation sizes, and relaxing $p$-value thresholds away from $p = 0.05$. 
\end{abstract}


\section{Introduction}

Experimentation in the private sector -- particularly in the online tech industry --  broadly differs from experimentation in scientific fields in two ways. Firstly, the goal of the vast majority of experiments in the private sector is to \textit{optimize} a metric. For example, tech firms optimize  measurable business objectives such as long-term retention, sign-up rates, or revenue. Secondly, repositories of past experiments give firms strong priors on treatment effects from related experiments. Empirical evidence \cite{azevedo2020b, goldberg2017decision, guo2020empirical, georgiev_1001, georgiev2022analysis, kohavi2024false} suggests the following stylized fact regarding experiments: treatment effect distributions are typically centered around zero, with significant dispersion around the mean. 
Figure \ref{fig:prior_densities} shows treatment effect densities for three distinct experimentation initiatives at Netflix, corroborating this point.

Given this prior on treatment effects, how should an organization conduct experimentation when the goal is to optimize a metric of interest? The current statistical paradigm in industry, null hypothesis statistical testing, can be mismatched with this goal: it is indifferent to utilities and opportunity costs, and it provides no principled justification for Type I and Type II error defaults \cite{manski2019treatment}. Building on a growing literature \cite{azevedo2020b,azevedo2023b, guo2020empirical, goldberg2017decision, deng2015objective}, we study an alternative Bayesian decision-theoretic framework and discuss how organizations should restructure current practice towards optimizing returns from experimentation.

\begin{figure}
    \centering
    \includegraphics[width=0.6\linewidth]{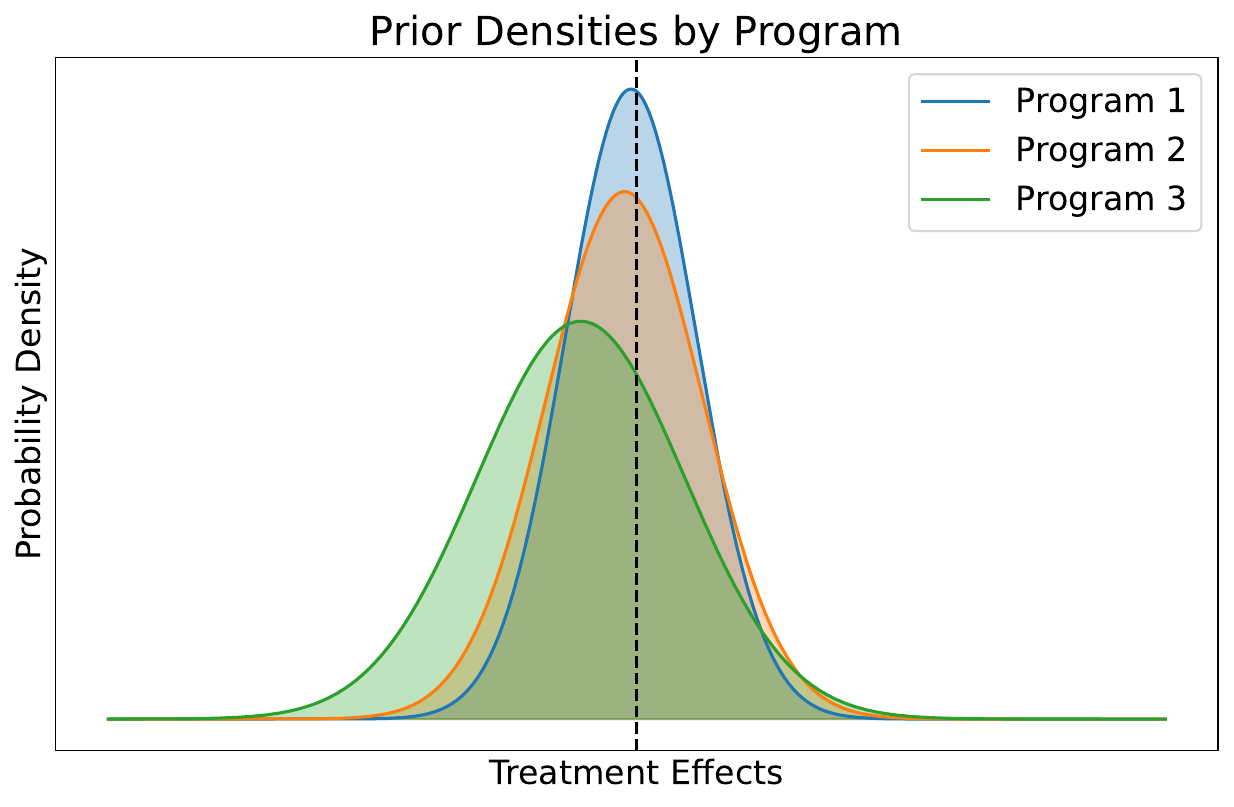}
    \caption{Best-fitting Gaussian densities for three programs at Netflix. $x$-axis is obfuscated. Dashed line indicates $x = 0$.} 
    \label{fig:prior_densities}
\end{figure}

\subsection{Setup}
\label{sec:setup}

An organization conducting experimentation typically has a pool of untested ideas and limited resources with which to conduct experiments. These organizations have a population of units on which to experiment -- an \textit{allocation pool} -- and teams to generate new ideas for experiments and subsequently analyze them. Generally, experimentation in these firms proceeds as follows. A product team generates an idea which may improve a product. The idea is tested using a randomized controlled experiment or A/B test \cite{kohavi2009online}. Based on the results of the A/B test, the idea is either discarded or \textit{shipped} -- made available to all units in the allocation pool. Often, similar tests on the same product, platform, or population can be grouped into initiatives that we will call \textit{experimentation programs}.

The \textit{A/B Testing Problem} \cite{azevedo2020b} provides a starting point for optimizing the returns from a single experimentation program. It is defined as follows. Suppose the experimentation program consists of $I$ untested ideas or interventions indexed by $i=1,\dots,I$. The data is not assumed to arrive sequentially. Treatment effects for these ideas are measured in the same metric. We will think of these effects $\Delta_i$ as the \textit{return} associated with idea $i$. Data from past experiments give credence to a Bayesian assumption that $\Delta_i$ are drawn iid for some distribution $G$. We choose an \textit{allocation} of units to A/B tests, $(n_i)_{i=1}^I$, such that $\sum_i n_i = N$, the total size of the allocation pool. A unit may only be assigned to one test at a time, making the allocation pool a limited resource. 

For those tests with $n_i > 0$, we observe the result of the experiment, given by
\[
\hat{\Delta}_i \sim \textsf{N}\left(\Delta_i, \frac{\sigma^2}{n_i} \right).
\]
Normality of the observations is justified by the central limit theorem; $\sigma$ represents typical fluctuations of the return metric, which is fairly stable and known in advance. The goal of the A/B testing problem is to select an allocation $(n_i)_{i=1}^I$ and a subset of tests $S$ to ship, $S \subseteq [I]$, in order to maximize total expected returns to experimentation,
\begin{equation}
\label{eq:ab_testing_problem_obj}
\mathbb{E}\left[\sum_{i \in S} u(\Delta_i) \right],
\end{equation}
where $u$ is an increasing utility function. We will typically take $u(x) = x$ to be risk-neutral. However $u$ may also capture risk-averse utilities such as $u(x) = x + bx\mathbf{1}\set{x < 0}, b > 0$, inspired by prospect theory \cite{kahneman2013prospect}.

\subsection{Related work}

The idea of optimizing returns from a portfolio of A/B tests is a stark alternative to traditional statistical analysis of experiments. The decision theoretic foundations of this perspective were first discussed in \cite{azevedo2020b, azevedo2023b, goldberg2017decision, manski2019treatment, guo2020empirical}. These works also adopt a Bayesian or empirical Bayesian approach to the problem. See Appendix \ref{sec:litreview} for an extended discussion of closely related work in the analysis of collections of A/B tests.

Most relevant to our work is the A/B Testing Problem in \cite{azevedo2020b}, which is the first to recast experimentation as a constrained optimization problem -- exactly the A/B Testing Problem in the previous section. By a careful asymptotic analysis, \cite{azevedo2020b} characterizes the optimal experimentation strategy in specific settings. For $t$-distributed priors which are sufficiently heavy-tailed, the optimal allocation is to split $N$ units equally across the $I$ tests, a strategy called \textit{lean experimentation}: firms should run many low-powered tests. For light-tailed $t$-distributed priors $G$ with $N$ small enough, the optimal strategy is to \textit{go big}: it is optimal to allocate all $N$ units to one test. 

Collectively, the aforementioned literature \cite{azevedo2020b, azevedo2023b, goldberg2017decision, manski2019treatment, guo2020empirical} and this paper represent a different approach to experimentation that we will refer to as the \textbf{return-aware framework}, to contrast it with null hypothesis testing. The framework views experimentation as a decision theoretic problem where the goal is to maximize the cumulative impact of experimentation, subject to various constraints.

\subsection{Contributions.} 

This article makes original contributions to the return-aware framework and discusses its implications for current practice. We also argue that the return-aware framework is well-suited to experimentation in industry, and worthy of future research. First, we highlight several simple, useful observations about the framework.

\paragraph*{Practical Solution Methods.} The elegant results of \cite{azevedo2020b,azevedo2023b} leave open the optimal allocation for general priors and arbitrary values of $I,N$. We show that dynamic programming solves this problem for any choice of $I,N$, utility function $u$, and prior $G$.
    
\paragraph*{Connection to $p$-value Frameworks.} In current experimentation practice, A/B tests are usually decided based on $p$-values and not posterior expectation calculations. We show that tuning the hypothesis testing level $\alpha$ based on the prior $G$ can replicate the optimal solution of Eq. \eqref{eq:ab_testing_problem_obj}. As a result, there is little engineering work required to adopt the return-maximizing framework, and the implied optimal $p$-value gives insight into how current practice should be changed.

\paragraph*{Analytic Approximations to the Optimal Solution.} We give a simple approximation to optimal returns from an experimentation program using the optimal allocation and decision rule. The approximation shows how program-level returns depends on the size of the allocation pool, the number of ideas available for testing, and properties of the prior distribution.

\paragraph*{Implications for Current Experimentation Practice.} 
Second, we apply the previous three insights to experimentation programs at Netflix, a large online technology firm. Under the assumption of no costs, they suggest that current practice, justified by null hypothesis statistical testing, is highly suboptimal for the goal of optimizing returns. The optimal allocation of the A/B testing problem suggests many more tests should be run in practice. Moreover, the implied optimal one-sided $p$-values are found to be markedly larger than $0.05$. Incorporating costs gives a more nuanced picture that we discuss in Section \ref{sec:costs}. Given the similarities in experimentation practice across technology firms, these results are likely applicable to other firms as well.
    
\paragraph*{Managing Multiple Experimentation Programs and Further Variants.} Thirdly, we discuss various extensions to the base A/B testing problem framework to managing multiple experimentation programs with shared resources. The extensions can be solved with dynamic programming and provide useful tools for managers to make investment decisions across several initiatives. \\

Finally, we end with important caveats on the return-aware framework and several directions for future research.

\section{Insights on the A/B Testing Problem}
\label{sec:ab_testing_problems}

The basic solution to the A/B Testing Problem is given in \cite{azevedo2020b}, which discusses both the optimal decision rule and resulting optimization problem describing the optimal allocation. We will focus exclusively on the case where $\E_G[\Delta] \leq 0$, which is a realistic assumption across many mature experimentation programs where most low hanging fruit has been picked \cite{azevedo2020b}. 

\begin{lemma}[Thm. 1 of \cite{azevedo2020b}]
\label{lemma:fat_tails}
\sloppy Assume that $(n_i)_{i=I_0+1}^I = 0$, and suppose the subset $S$ is chosen on the basis of the observed $(\hat{\Delta}_i)_{i=0}^{I_0}$. Given any allocation strategy, the optimal subset $S$ is comprised of $i \leq I_0$ such that the posterior expectation $\E[u(\Delta_i) | \hat{\Delta}_i]$ is positive. 
\end{lemma}

\sloppy Given the optimal decision, the allocation problem then reduces to maximizing 
\[
\sum_{i=1}^I \mathbb{E}\left[\mathbb{E}\left[ u(\Delta_i) | \hat{\Delta}_i;n_i \right]^+ \right]
\]
subject to a simplex constraint. The summand 
\[
f(n) := \mathbb{E}\left[\mathbb{E}\left[ u(\Delta) | \hat{\Delta};n \right]^+ \right]
\]
as a function of $n$ is called the \textit{production function} \cite{azevedo2020b}, which represents the expected return from testing a typical idea with $n$ units, using the optimal decision. When $n = 0$, we set $f(n) = 0.$ The production function can be computed with knowledge of the treatment effect distribution $G$ and $\sigma$, both of which we assume are known. In practice, $G$ can be estimated using data on past experiments, using methods described in Sec. \ref{sec:litreview}; see also Section \ref{sec:discussion}. \\

The allocation problem thus reduces to
\begin{align*}
    & \text{maximize}_{(n_i)_{i=1}^{I}} \sum_{i=1}^I f(n_i)\\
    & \text{subject to } n_i \geq 0, \sum_{i=1}^{I} n_i = N.
\end{align*}
We show that the problem can be solved in general with simple dynamic programming. Denote by $F(I,N)$ the maximum such value.  Notice that 
\begin{equation}
\label{eq:DP_eq}
F(I,N) = \max_{j = 0}^N \left( F(I-1,N-j) + f(j)\right), 
\end{equation}
which suggests an efficient approach to calculate $F(I,N)$ by tabulating values $F(i,n)$ for all $i=1,\dots,I, n = 1,\dots,N$. The base cases are simple; when $I = 1$, we just return the value $f(N)$. When $N = 1$, we return the value $f(1) + (I-1)f(0).$ The dynamic programming algorithm takes $O(N^2 I)$ time after the production function is computed.  Computing the optimal solution can be done by backtracking, keeping track of the index $j$ which maximizes $F(I,N)$ in equation \eqref{eq:DP_eq}. If $N$ is prohibitively large, we can consider a minimum experiment cohort size $c_0$ and instead use the update
\begin{equation}
\label{eq:DP_eq_discretized}
F(I,N) = \max_{j = 0}^{\floor{N/c_0}} \left( F(I-1,N-jc_0) + f(jc_0)\right),    
\end{equation}
which requires $O((N/c_0)^2 I)$ time to solve. These methods can be used to solve many related extensions of the return-maximizing framework, as we show throughout the paper. Code implementing the dynamic programming methods may be found on \href{https://github.com/tsudijon/OptimizeExperimentPrograms}{Github}. 
The repository contains notebooks to replicate all figures in the paper and also an example dataset of treatment effects from a Netflix experimentation program, scaled by a random factor to preserve confidentiality.

\subsection{Connection to $p$-value Decision Frameworks}
\label{sec:pvalue_framework}

At many online technology firms, the traditional framework of $p$-value significance is used; tests are rejected at the $\alpha = 0.05$ level using covariated-adjusted causal estimators, with sample size allocated by power calculations. One-sided $p$-values for the null hypothesis $H_0: \Delta_i \leq 0$ are used to determine whether idea $i$ should be shipped. In the model described in Section \ref{sec:setup}, the one-sided $p$-value is given by 
\[
p_i := 1-\Phi(\widehat{\Delta}_i\sqrt{n_i}/\sigma).
\]
Perhaps surprisingly, decisions based on this one-sided $p$-value are equivalent to the optimal decisions in the A/B Testing Problem based on the posterior expectation.

\begin{proposition}
\label{prop:pvalue_framework_equivalence}
There exists an $\alpha$ threshold such that 
\begin{equation}
p_i \leq \alpha \Leftrightarrow \E_G[u(\Delta_i) | \widehat{\Delta}_i] > 0,
\end{equation}
for any prior $G$ and increasing utility function $u.$
\end{proposition}
A proof of this result and all others can be found in the Appendix. Prop. \ref{prop:pvalue_framework_equivalence} suggests there is no loss of generality in considering a $p$-value decision framework versus one which ships when the conditional expectation is positive. Only the $\alpha$ budget per test or equivalently the $t$-statistic threshold needs to be calibrated.

As a result, organizations can leverage existing infrastructure which already calculates $p$-values. Further, the mapping between the A/B Testing Problem and $p$-values suggests an optimal choice for $p$-value thresholds that maximize returns. We explore this further in Section \ref{sec:suboptimality_decisions}.

\begin{example}[Normal Prior, Linear Utility]
\label{ex:normalprior_linearutility}
In the normal prior case, the $t$-statistic thresholds are explicit. Take $u(x) = x$. In this case, we begin by finding the critical value $c_x$ such that $\E[\Delta_i | \hat\Delta_i = c_x] = 0$. Recall that the posterior distribution of $\Delta_i$ given $\widehat{\Delta}_i$ is $\sf{N}(m,s^2)$ with 
\begin{align*}
    m & := \hat \Delta_i \frac{\tau^2}{\tau^2 + \sigma^2/n_i} + \mu\frac{\sigma^2/n_i}{\tau^2 + \sigma^2/n_i} \\
    s^2 & := \frac{\tau^2\frac{\sigma^2}{n}}{\tau^2 + \frac{\sigma^2}{n}}.
\end{align*}
Inspecting the posterior mean, we find $c_x = -\mu\sigma^2/n_i\tau^2$. The resulting threshold as a $t$-statistic is given by

\begin{equation}
\label{eq:optimal_threshold}
\frac{-\mu\sigma}{\tau^2\sqrt{n_i}},
\end{equation}

which is dependent on the test, and can be interpreted as the product $\mu/\tau$ and the precision ratio $\frac{\sigma/\sqrt{n_i}}{\tau}.$ It is not difficult to implement the proper $\alpha$ threshold corresponding to \eqref{eq:optimal_threshold}. For $\mu = 0$, this would correspond to rejecting whenever the one-sided $p$-value is smaller than $\alpha = 1/2.$ 

\begin{remark}
The same formula for the optimal $t$-statistic threshold is contained in Theorem 2 of \cite{goldberg2017decision} and also Proposition 1 of \cite{azevedo2023b}. We show that the connection is generic, in the sense that a properly tuned $t$-threshold can always match the optimal solution based on posterior expectations, for any increasing utility $u$ and prior $G.$

Figure \ref{fig:test_passing} shows various statistics of the optimal decision rule as a function of $\sqrt{n_i}.$ As the number of units allocated to a test increases, the test stringency decreases. This could result in a misalignment of incentives. Individual experimenters desire to make $n$ as large as possible to make the test easier to pass, while the return-aware framework framework assumes that allocating units to A/B tests is done on the program level.

A simple computation shows that the probability of a test passing is given by the marginal probability that $\hat \Delta_i$ is greater than $-\mu\sigma^2/\tau^2n$. The probability is thus 
    \[
    \Phi\left(\frac{\mu}{\tau^2} \sqrt{\tau^2 + \sigma^2/n} \right).
    \]
Figure \ref{fig:test_passing} shows that for the parameters estimated at Netflix, the probability of a test passing rapidly asymptotes to $\Phi(\mu/\tau)$, so there is a weak incentive to overallocate to a test. 
\end{remark}

\begin{figure}
    \centering
    \includegraphics[width = 0.7\linewidth]{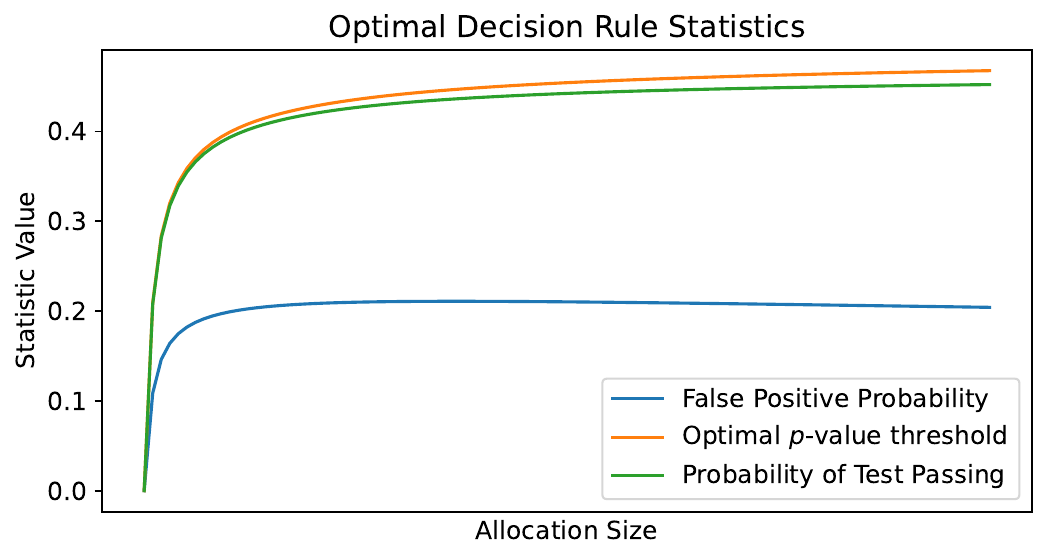}
    \caption{$p$-value threshold as a function of allocation for the optimal $t$-statistic threshold in Equation \eqref{eq:optimal_threshold}, alongside other statistics for a particular experimentation program of interest. $x$-axis is hidden to not disclose specifics of Netflix experimentation practice but is on the order of $10^5$.}
    \label{fig:test_passing}
\end{figure}

\end{example}

\subsection{Approximating Optimal Program-Level Returns}
\label{sec:metaproduction}


The dynamic programming solution suggests that the optimal allocation distributes units equally across some subset of tests $I_0 \leq I$ and zero on the rest. This suggests a useful approximation for the optimal value of the A/B testing problem, given by
\begin{equation}
\label{eq:metaproduction}
F(I,N) := \max_{i \in [1,I]} i \; f\left( \frac{N}{i}\right),
\end{equation}
where $f$ is the production function in the base A/B testing problem. We will refer to this as the \textit{metaproduction function}, representing expected program-level returns using the optimal allocation and decision procedure. Studying $F$ lets us understand the sensitivity of program returns to $I,N$ and various features of the prior. For example, consider the \textit{idea meta-production function} $F_1(I) = F(I,N)$ for a fixed large $N$. The shape of $F_1$ sheds insight into the value of idea generation, which can inform investment decisions into the experimentation program. 

To gain intuition on $F$, we will focus on the setting where $G$ is Gaussian with $u(x) = x.$ By conjugacy, the production function can be explicitly calculated. This formula was first written in \cite{raiffa2000applied} and recently studied in \cite{azevedo2023b}; similar formulas appear in \cite{abadie2023estimating}. 

\begin{lemma}[Proposition 2 of \cite{azevedo2023b}]
Defining $v(n) = \tau^2 + \sigma^2/n$, then $f(n)$ is given by
\begin{equation}
\label{eq:normal_production_function}
 \frac{\tau^2}{\sqrt{v(n)}} \phi\left(\frac{\mu}{\tau^2}\sqrt{v(n)}\right)  + \mu\Phi\left(\frac{\mu}{\tau^2}\sqrt{v(n)}\right)
\end{equation}
\end{lemma}

This can be simplified to $\mu G\left(\frac{\mu}{\tau^2}\sqrt{\tau^2 + \frac{\sigma^2}{n}} \right)$ where $G(x) = \phi(x) / x + \Phi(x).$ It is shown in \cite{azevedo2023b} that the production function is generically convex on some interval $[0,\hat{x}]$ and concave afterwards. The value $\hat{x}$ depends on $\mu,\tau$. The following theorem describes $F$ for a general Gaussian prior.

\begin{theorem}[Description of the Metaproduction Function]
\label{thm:metaprod_description}
Let $x^*$ be the global maximum of the function $f(x)/x$ on $\mathbb{R}^+$. Suppose $I \geq 1$. Then 
\begin{equation}
\label{eq:metaprod_description}
F(I,N) = 
\begin{cases}
    f(N) & \text{if } x^* \geq N \\
    N\frac{f(x^*)}{x^*} & \text{if } x^* \in [\frac{N}{I},N] \\
    I \; f\left( \frac{N}{I}\right) & \text{if } x^* \leq \frac{N}{I}.
\end{cases}
\end{equation}
Moreover, the number of tests maximizing $if(N/i)$ is given by
\begin{equation}
\label{eq:metaprod_argmax_description}
\argmax_{[1,I]} \; \; i \;f\left(N/i \right) = 
\begin{cases}
    1 & \text{if } x^* \geq N \\
    N/x^* & \text{if } x^* \in [\frac{N}{I},N] \\
    I & \text{if } x^* \leq \frac{N}{I}.
\end{cases}
\end{equation}
For $I \leq 1$, set $F := 0.$
\end{theorem}

Theorem \ref{thm:metaprod_description} yields important insights about the value of additional members and ideas for experimentation programs; implications are given in the following corollaries and in Section \ref{sec:metaproduction_insights}. The proof also shows that $x^* \geq \hat x$. Furthermore, the theorem only uses the fact that $f(x)$ is convex for small $x$ then concave for large $x$; if this holds for other priors, then the same results follow. 

\begin{corollary}[Metaproduction Function in $N,I$] 
\label{cor:metaproduction_marginals}
With the same notation as the previous theorem, the following statements are true.
\begin{enumerate}
    \item As a function on $[1,I]$, $F_1(I)$ is concave increasing and is constant for $I$ large enough.
    \item $F_2(N)$ is increasing, convex on interval $[0,\hat{x}]$, and concave on $[\hat x, \infty)$. Moreover, it is linear in $N$ for some intermediate range, strictly concave for large enough $N$, and asymptotes to a fixed value.
\end{enumerate}
\end{corollary}

Understanding the production function also leads to bounds on a type of regret for the base A/B testing problem. This is similar to results quantifying the value of experimentation in \cite{abadie2023estimating}.

\begin{corollary}[Comparing the Full Oracle and Oracle Bayes Solution]
\label{cor:regret_with_oracle}
Let $L_1$ be the expected difference between the full oracle decision and the oracle Bayes decision:
\[
\E\left[ \sum_i \Delta_i^+ \right] - \E\left[ \sum_i \E[\Delta_i|\hat{\Delta}_i]^+ \right].
\]
Suppose $N,I \rightarrow \infty$ such that $N/I \rightarrow \kappa \in (0,\infty)$. Then
\[
\lim_{I \rightarrow \infty} \frac{L_1}{I} =  
\begin{cases}
\E\left[\Delta_i^+\right] - \kappa f(x^*)/x^* & \text{if } \kappa \leq x^* \\
\E\left[\Delta_i^+\right] - f(\kappa) & \text{if } \kappa > x^*
\end{cases}.
\]
\end{corollary}
Thus, the loss scales with the number of ideas, and the loss per idea depends on the ratio $\kappa$. Since $\lim_{\kappa \rightarrow \infty} f(\kappa) = \E[\Delta_i^+],$ for large enough $N$ compared to $I$, the loss per idea is negligible.


Finally, we analyze how the changes to the prior density of ideas affects $F$. This shows the value of changing the distribution of ideas as a lever for optimization experimentation returns and extends "comparative statics" results found in \cite{azevedo2023b}.

\begin{corollary}[Metaproduction Function in $\mu,\tau$]
\label{cor:metaproduction_vs_mu_tau}
For any fixed $I,N$, $F(I,N)$ is increasing as $|\mu| \rightarrow 0$ for fixed $\tau$; it is also increasing as $\tau \rightarrow \infty$ for fixed $\mu$.
\end{corollary}

A quick inspection of Eq. \eqref{eq:normal_production_function} shows that $f(n) \rightarrow \infty$ as $\tau \rightarrow \infty.$ Thus if $\tau$ increases quickly enough, even as $\mu \downarrow \infty$, it is possible to obtain more value from an experimentation program.

\section{Implications on Experimentation Practice}
\label{sec:implications}

In many firms, experimentation programs are synonymous with business departments; the managers of these business departments use experimentation as a tool for decision making and optimizing the growth of the program. We highlight the implications of the return-aware framework by answering the following questions that these managers often have in practice:

\begin{enumerate}
    \item How many tests should we be running with a given pool of ideas and allocation? Are many underpowered tests better than a few high-powered tests? 
    \item Would changing $p$-value thresholds increase experimentation returns? Should other decision rules be used?
    \item What is the value of generating more ideas for a particular experimentation program? What is the impact of changing the underlying distribution of idea returns? 
    \item How should limited resources be allocated to multiple experimentation programs?
\end{enumerate}

\begin{figure}
    \centering
    \includegraphics[width=0.7\linewidth]{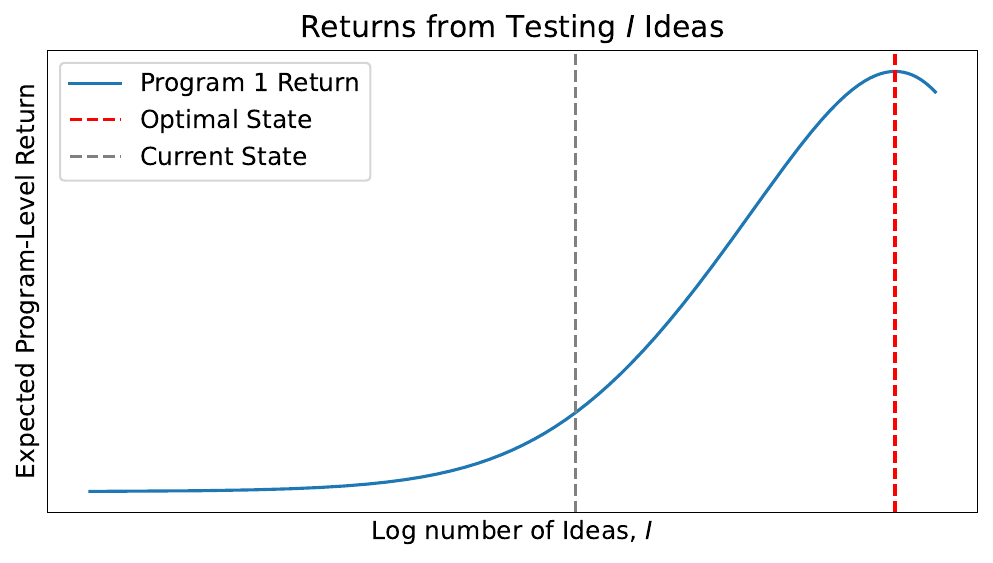}
    \caption{Blue curve shows expected return to testing $I$ ideas with equally sized allocations, using an allocation pool of 250 million units and prior parameters estimated from Program 1 of Fig. \ref{fig:prior_densities}. $x$-axis scale is hidden but spans 4 orders of magnitude. Grey dashed line denotes the number of ideas tested currently in practice, while red dashed line indicates optimum of the return curve.}
    \label{fig:value_of_testing}
\end{figure}

We will answer these questions \textit{under the assumption of no costs}. Section \ref{sec:costs} discusses the role of costs as well as risk-aversion, adding significant nuance to the takeaways of the return-aware framework.

\subsection{Allocations \& The Scale of Testing}
\label{sec:allocation_practice}

This section addresses Question 1. Figure \ref{fig:value_of_testing} shows the value of testing $I$ ideas with a fixed allocation pool of 250 million units, split equally amongst the ideas. The $x$-axis spans 4 orders of magnitude; the red dashed line shows the optimum of the curve, and the grey dashed line shows roughly the scale of testing in current practice. With the available allocation and a much larger pool of untested ideas $I$, it would be optimal to run several orders of magnitude more tests of typical quality than current practice. For any value of $I$ to the left of the optimum, Figure \ref{fig:value_of_testing} also 
suggests that we should test all ideas with equal allocation. The resulting test sizes are much smaller than current test sizes, which are on the order of several million and obtained via power calculations.

In this costless regime, many more tests should be run at the expense of higher allocation and power, which is the strategy termed \textit{lean experimentation} in \cite{azevedo2020b, azevedo2023b}. The surprising finding is that the strategy is optimal for a wide range of parameters, even with a ``light-tailed" Gaussian prior, which gives the opposite intuition from \cite{azevedo2020b}; this nuance was pointed out in \cite{azevedo2023b}. Interestingly, lean experimentation is closely related to the \textit{paradox of power} described in \cite{schmit2019optimal}. \cite{schmit2019optimal} considers a sequential model of testing where the goal is to find discoveries as quickly as possible; the optimal strategy is close to running many short, low-powered tests in sequence.

\subsection{Suboptimality of Standard Decision Rules}
\label{sec:suboptimality_decisions}
This section answers Question 2 by studying the optimal decision rule from Lemma \ref{lemma:fat_tails}. Throughout this section, we will consider decision rules for a test with allocation size $5$ million and assume treatment effects come from the Gaussian prior estimated for Program 1 of Figure \ref{fig:prior_densities}. Figure \ref{fig:test_passing} displays several statistics on the optimal decision rule for a particular experimentation program of interest. A striking feature is that the $p$-value threshold is closer to $0.5$ than to $0.05.$ See Section \ref{sec:new_exp_programs} for a discussion on minimax approaches that suggest similar $p$-value thresholds. This recommendation aligns with several other papers, which also show that $p < 0.05$ decision rules are too conservative \cite{goldberg2017decision, azevedo2020b}.

We will reinforce this finding by demonstrating the suboptimality of $p < 0.05$ decision results in various settings. Given a linear utility with no costs, one way to demonstrate subtoptimality is to visualize other Gaussian priors that induce an optimal decision rule where $p < 0.05$. By Equation \ref{eq:optimal_threshold}, we display several priors whose parameters $\mu,\sigma$ satisfy $-\mu \sigma/\sqrt{n} = \Phi^{-1}\left( 0.95\right) \tau^2.$  Figure \ref{fig:p0.05optimal_priors} shows several such distributions, superimposed on the actual estimated prior. Using the $p < 0.05$ rule corresponds to an unnecessarily pessimistic prior where there is little mass above $x > 0$. The results also show that the general recommendation of adopting more conservative decision thresholds is robust to mis-estimation of the prior distribution.

\begin{figure}
    \centering
    \includegraphics[width=0.6\linewidth]{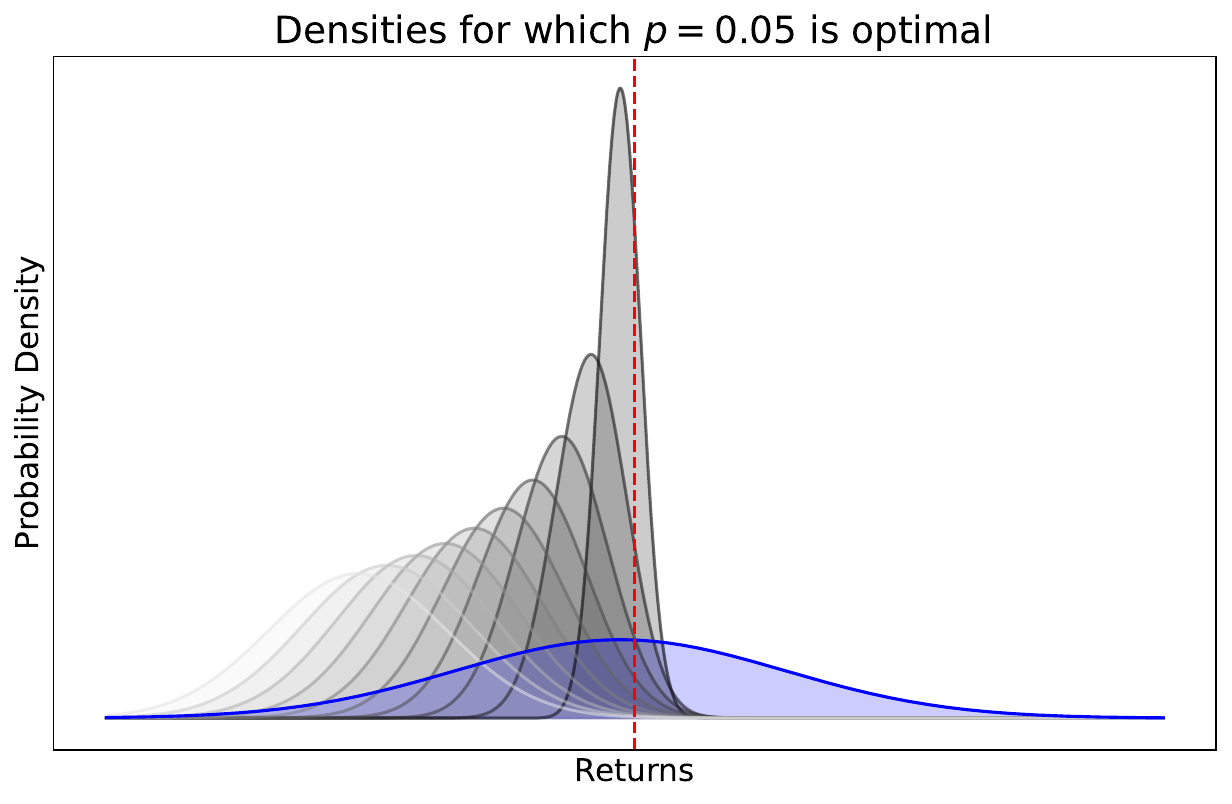}
    \caption{Example Gaussian densities for which the optimal decision for maximizing returns is to ship when the usual one-sided $p$-value is less than $0.05$, shaded in gray. The blue curve represents the Gaussian prior for Program 1 of Figure \ref{fig:prior_densities}. The $x$-axis is obfuscated.}
    \label{fig:p0.05optimal_priors}
\end{figure}

We may also calculate the average return lost when using the suboptimal $p < 0.05$ decision rule. Figure \ref{fig:p0.05_comparison} compares two decision rules, (1) shipping when $\E[\Delta | \widehat{\Delta}] > 0$ and (2) shipping when $p < 0.05$, using the prior $G.$ In the top plot, we compare the expected return from optimally allocating $N$ units to $I$ ideas, with both decision rules. For (1), this is exactly the metaproduction function, while for (2), we obtain the analog of the metaproduction function, using instead $f_{p < 0.05}(n) = \E\left[\Delta_i\mathbf{1}\set{\widehat{\Delta}_i \geq \frac{1.96\sigma}{\sqrt{n}}}\right].$ The simulated example shows that using $p < 0.05$ decisions can miss out on nearly 80\% of returns, as compared to the optimal decision strategy, depending on $I,N$. The bottom plot compares the production function $f(N)$ and $f_{p < 0.05}(N).$ 


\begin{figure}
    \centering
    \includegraphics[width = \linewidth]{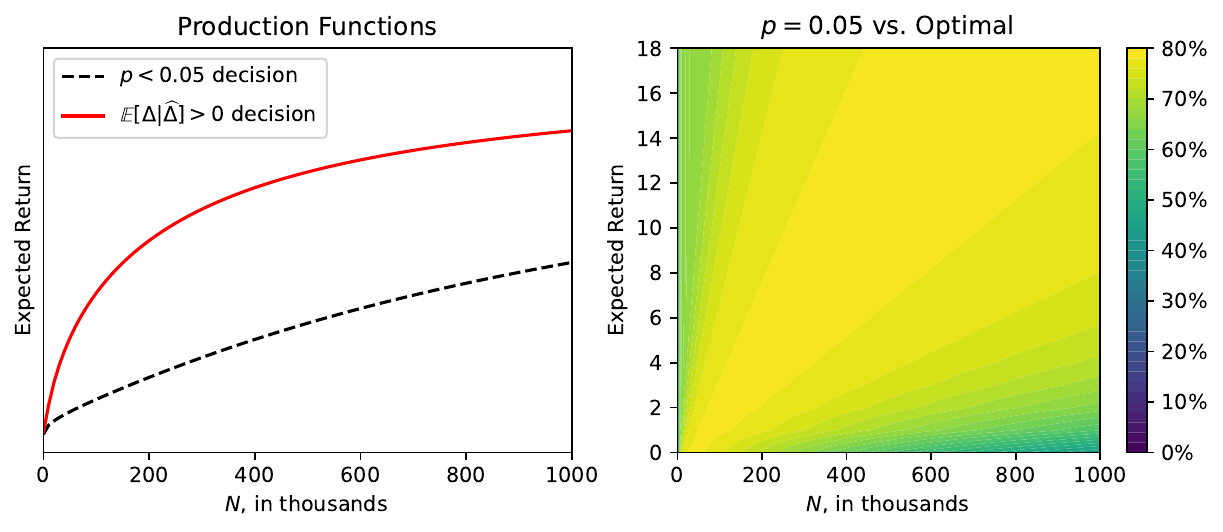}
    \caption{Left: expected return from allocating $N$ units to test a single idea drawn from a prior $G$, using the two different decision procedures. Right: comparison of the expected return from optimally allocating $N$ units to $I$ ideas, shipping ideas based on (1) posterior expectation greater than zero, and (2) $p < 0.05$. The heatmap shows the percentage of return lost from using (2) compared to (1).  Both figures use a Gaussian prior $G$ with parameters similar to those in real experimentation programs. }
    \label{fig:p0.05_comparison}
\end{figure}

\subsection{Insights from the Metaproduction Function.}
\label{sec:metaproduction_insights}

Studying the metaproduction function for an experimentation program answers Question 3. Figure \ref{fig:netflix_metaproduction} shows the metaproduction function with estimated Gaussian parameters on Program 1 of Figure \ref{fig:prior_densities}. The results show that increasing the allocation pool has rapidly diminishing returns, while adding more untested ideas into the pool remains valuable, with slowly diminishing marginal returns. The theory presented in Section \ref{sec:metaproduction} suggests that beyond a point, adding additional ideas to the pool does not increase program-level returns. The scale of current practice is nowhere near this breakpoint.

The results of Section \ref{sec:allocation_practice} suggest the need to generating many more ideas, given that idea quality remains the same. The right panel of Figure \ref{fig:netflix_metaproduction} demonstrates the value of doing so, with a fixed allocation pool of 250 million units. With additional information on the costs of idea generation, we may use the metaproduction function $F$ to determine the size of investment into the program, in terms of the number of untested ideas available. 


The theory in Section \ref{sec:metaproduction} also suggests a few insights. The first is intuitive: organizations should increase idea average quality, as Corollary \ref{cor:metaproduction_vs_mu_tau} shows. If this is difficult, an alternative is to increase the \textit{diversity} of ideas, corresponding to increasing $\tau$. Secondly, Theorem \ref{thm:metaprod_description} suggests under which conditions the lean experimentation is optimal: organizations should test all ideas in the pool if $I \leq N/x^*$ where $x^*$ is the global maximum of $f(x)/x$ where $f$ is the production function of the program.

\begin{figure*}
    \centering
    \includegraphics[width = 1.01\textwidth]{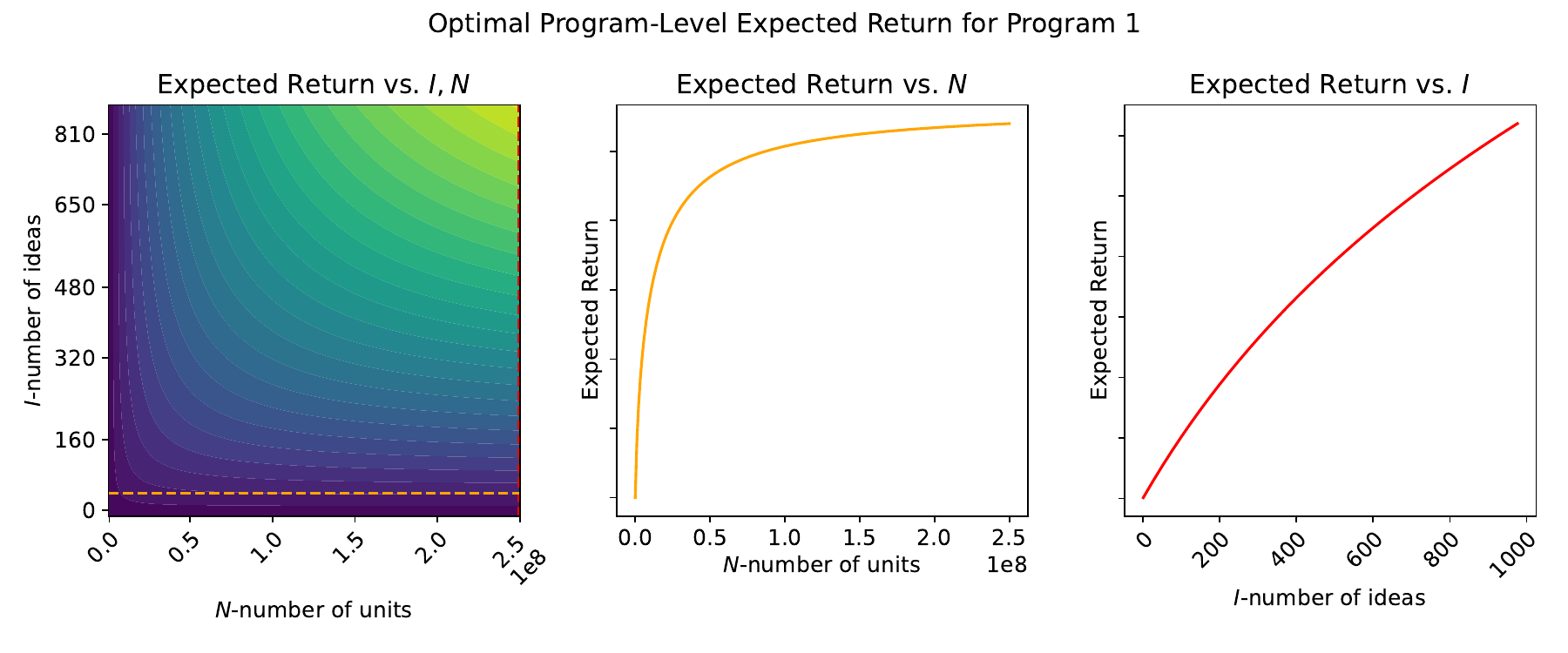}
    \caption{
    Visualization of the metaproduction function using a Gaussian prior, estimated using data from a Netflix experimentation program. Lighter colors indicate larger returns. Left: the two dimensional metaproduction function $F(I,N)$. Middle: one-dimensional projection of $F$ where $I = 50$ and $N$ increases, corresponding to the dashed line on the left plot. Right: one dimensional slice of $F$ with $N = 10^8$.}
    \label{fig:netflix_metaproduction}
\end{figure*}

\subsection{The Role of Costs}
\label{sec:costs}

Our discussion of optimizing returns has only touched on the role of costs.  If the implementation cost for idea $i$ is $s_i$ and the testing cost is $t_i(n_i)$, as a function of the allocation $n_i$, then the production function changes and is given by 
\[
\E\left[(\E[\Delta_i | \widehat{\Delta}_i; n_i] - s_i)^+ \right] - t_i(n_i).
\]
Recall that the production function measures the expected return from allocating $n_i$ units to test a typical idea $i$, using the optimal decision. Here, the optimal strategy is to ship all ideas with posterior expectation greater than the implementation cost $s_i$. The allocation problem reduces to maximizing a separable objective with a simplex constraint, so our computational techniques still apply. 

The difficulty with understanding the role of costs is that costs are typically difficult to measure. For example, there are hidden costs per test that present as accumulated complexity, when more and more treatments are incorporated into the platform. Measuring costs of human labor for collaborative projects is not straightforward. Furthermore, the metric of returns in most experiments is different from dollars, and converting between the two metrics often relies on a noisy black-box. As a result, the exact takeaway for experimentation practice depends on the firm. We suggest a suite of several exercises that firms can use to study decision procedures given their bespoke costs and risk tolerances, as follows. The exercises are illustrated on Program 1 of Figure \ref{fig:prior_densities}. See also \cite{azevedo2020b} for an excellent related discussion on data from Microsoft Bing.

We first study the dependence of $p$-value thresholds on implementation costs; we can reverse-engineer the implementation costs for which $p < 0.05$ is optimal given a linear utility and the estimated prior from Figure \ref{fig:prior_densities}. With cost $s_i$, it is optimal to ship a tested idea $i$ whenever
\[
\E[\Delta_i | \widehat{\Delta}_i] \geq  s_i.
\]
By the same arguments as Prop. \ref{prop:pvalue_framework_equivalence}, there is a unique value of $\widehat{\Delta}_i$ above which it is optimal to ship. Figure \ref{fig:backing_out_costs} shows the implied $p$-value threshold as a function of $s_i$. The result suggests that the $p = 0.05$ decision threshold is optimal with implementation costs of about 9 times the magnitude of $|\mu|$. Whether this cost is significant or insignificant greatly depends on the scale of $|\mu|$ and other parameters specific to the experimentation program. Interpreting this cost should be done in a firm-specific manner. However, rough conversions between this metric and dollars in the Netflix context suggests an implementation cost equivalent to a year's worth of work for several engineers or data scientists. This is realistic for some ideas, but not others.

\begin{figure}[htbp]
    \centering
    \begin{subfigure}{0.49\linewidth}
        \centering
        \includegraphics[width=\linewidth]{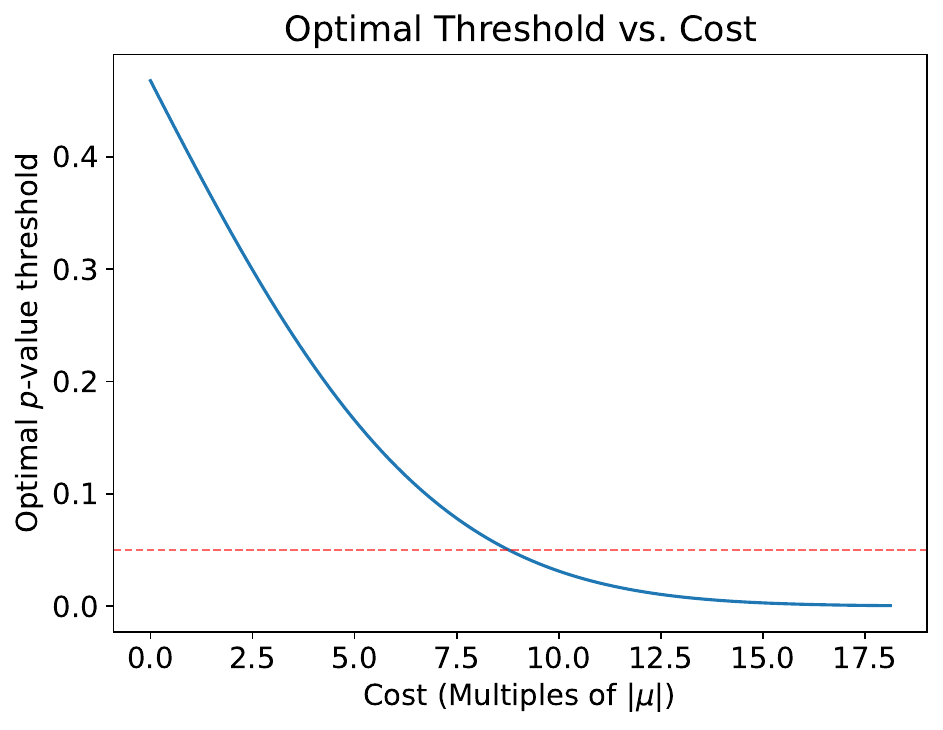}
        \caption{}
        \label{fig:backing_out_costs}
    \end{subfigure}
    \hfill
    \begin{subfigure}{0.49\linewidth}
        \centering
        \includegraphics[width = \linewidth]{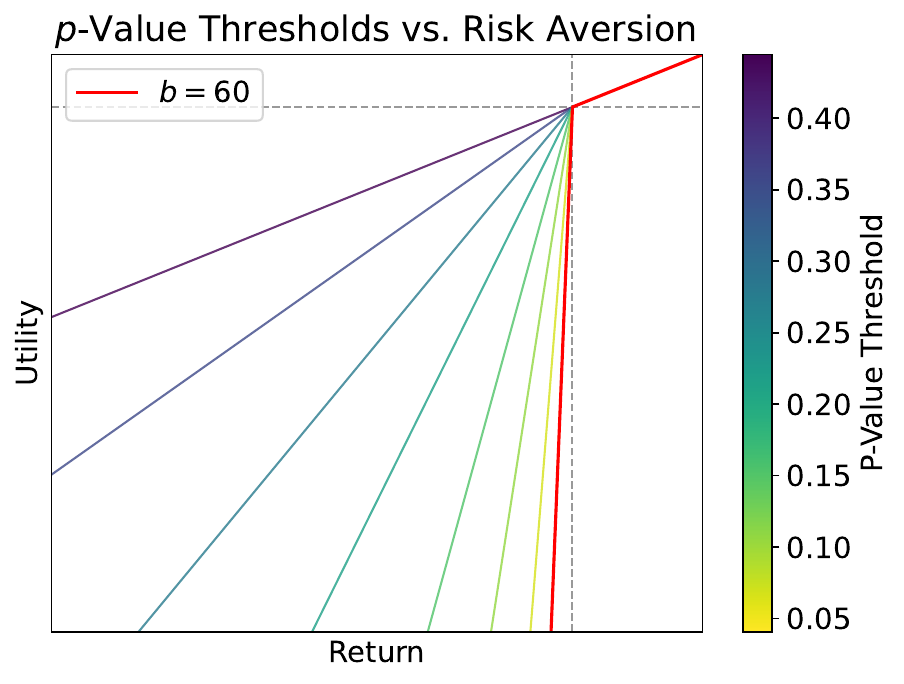}
        \caption{}
        \label{fig:backing_out_utilities}
    \end{subfigure}
    \caption{Left: Optimal $p$-value threshold as a function of costs $c_i$ for linear utility and estimated $G$ from Figure \ref{fig:prior_densities}. $x$-axis is scaled by the size of the prior mean. The red-dashed line indicates a $p = 0.05$. Right: Various risk averse utilities of the form $u(x) = x + bx\mathbf{1}\set{x < 0}, b \geq 0$ for various values of $b$, colored by the optimal $p$-value decision threshold. Red line denotes the utility function with optimal decision $0.05$, corresponding to $b \approx 60$. Gaussian prior $G$ from Program 1 of Fig. \ref{fig:prior_densities} was used, with $n = 5$ million.}
\end{figure}

Along similar lines, we can visualize several utility functions for which $p < 0.05$ is optimal. Consider $u(x) = x + bx\mathbf{1}\set{x < 0}, b \geq 0.$ $b = 0$ recovers the linear utility we have primarily focused on. The optimal decision given this utility is to ship all tests for which 
\begin{equation}
\label{eq:risk_averse_utility_decision}
    \E[\Delta_i + b \Delta_i \mathbf{1}\set{\Delta_i < 0} | \widehat{\Delta}_i = c] \geq 0
\end{equation}
By normal conjugacy, $\Delta_i | \widehat{\Delta}_i \sim \sf{N}(m,s^2)$ with posterior mean and variance given in Example \ref{ex:normalprior_linearutility}. Using integration by parts, \eqref{eq:risk_averse_utility_decision} reduces to solving for the value of $c$ which satisfies
\[
m\left( 1 + b\Phi\left(-\frac{m}{s}\right)\right) - bs\phi\left(\frac{m}{s}\right) = 0.
\]

Figure \ref{fig:backing_out_utilities} displays the resulting optimal $p$-value thresholds for various risk averse utilities with different choices of $b$, for Program 1 of Fig. \ref{fig:prior_densities}. The results suggest that a $p < 0.05$ decision is optimal for a risk-averse utility which weights losses about \textit{60 times more} than gains. The exact multiple varies depending on the properties of the prior distribution. When repeated for the other priors in Figure \ref{fig:prior_densities}, the implied risk-aversion still weights losses much more heavily than gains.

Finally, Figure \ref{fig:value_of_testing_with_costs} shows the impact of fixed costs per test on the optimal allocation in various scenarios. The same curve as in Figure \ref{fig:value_of_testing} is pictured in blue, while the other curves show the cumulative returns to testing $I$ ideas, given a fixed testing cost $t_i(n_i) = c$. We test four settings, corresponding to no cost, an unspecified economically significant cost $c$, and multiples $5c, 10c$. In each of these settings, the optimal allocation is several times larger or even orders of magnitude larger than current practice at Netflix.

\begin{figure}
    \centering
    \includegraphics[width=0.7\linewidth]{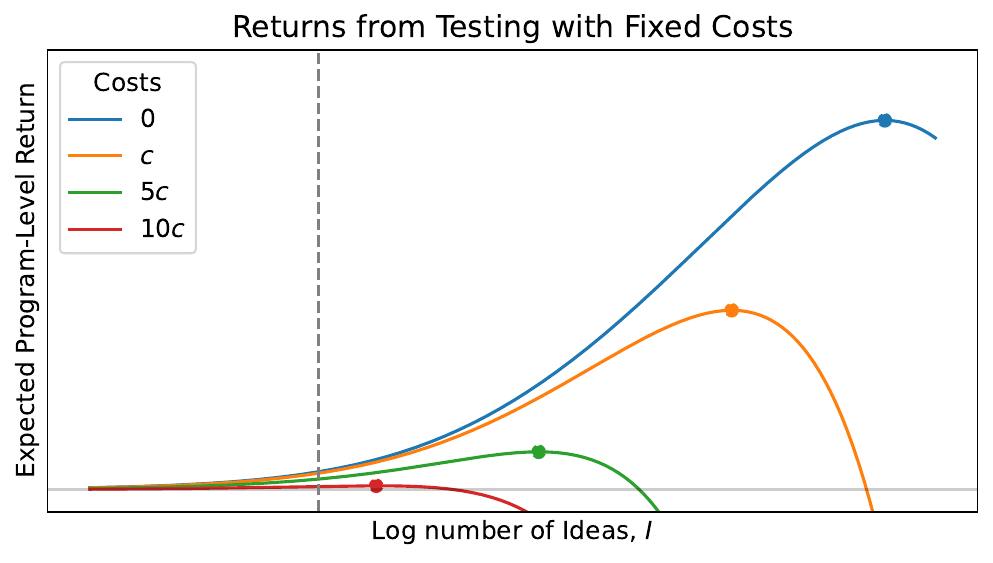}
    \caption{Expected return to testing $I$ ideas with equally sized allocations, using an allocation pool of 250 million units and prior parameters estimated from Program 1 of Fig. \ref{fig:prior_densities}, with per test costs. $x$-axis scale is hidden but spans 4 orders of magnitude. Grey dashed line denotes the number of ideas tested currently in practice.}
    \label{fig:value_of_testing_with_costs}
\end{figure}

The results of these exercises will result in different takeaways per firm. However, it is clear that when framing experimentation in the return-aware framework, one should both measure and incorporate costs into any practical analysis. Further research should be conducted on the sensitivity of these results to costs.

We also highlight the need for better infrastructure to measure costs and drive them down as much as possible. The current state of experimentation in online platforms generally suggests that costs to testing are rapidly decreasing.  At many online technology firms, there are centralized teams -- \textit{experimentation platforms} -- whose explicit mandate is to reduce the friction of running A/B tests; for smaller organizations without such teams, there are external vendors who provide experimentation services. One ultimate goal of these efforts is to bring experimentation into the costless regime, where the return-aware framework suggests significant changes to current practice. For example, in experimentation platforms, a common business question is to understand the value of investing a large capital investment to decrease fixed costs for future experiments. The return-aware framework can be used to quantify tradeoffs in case studies like these.


\subsection{Managing Multiple Experimentation Programs}
\label{sec:multiple_programs}

At many firms, experimentation is conducted in multiple areas which have different priors. Multiple experimentation programs may share constrained resources, including allocation or R\&D resources.
Addressing Question 4, we discuss an extension of the A/B testing problem to answer how allocation should be divided amongst the programs, in order to optimize returns. Answering this question provides useful insight for making investment decisions across multiple experimentation areas.

Suppose that there are $P$ experimentation programs indexed by $p=1,\dots,P$, each with $I_p$ tests, indexed by $i=1,\dots,I_p$.  Assume that test returns are drawn from distributions $G_p$ so that 
\begin{align*}
\Delta_{i,p} & \stackrel{\textsf{i.i.d.}}{\sim} G_p \\
\hat{\Delta}_{i,p} & \sim \sf{N}\left(\Delta_{i,p},\frac{\sigma^2_i}{n_{i,p}} \right).
\end{align*}

The return-aware framework is the same, except unit allocation to each program comes from a shared pool of $N$ units. This variant of the framework can be analyzed in the same way, resulting in the optimization problem 
\begin{align*}
    & \text{maximize}_{n_{i,p}} \sum_{p=1}^{P} \sum_{i=1}^{I_p} \mathbb{E}\left[\mathbb{E}\left[ \Delta_{i,p} | \hat{\Delta}_{i,p};n_{i,p} \right]^+ \right] \\
    & \text{subject to } n_{i,p} \geq 0, \sum_p\sum_{i=1}^{I_p} n_{i,p} = N.
\end{align*}
Understanding the solution of this problem requires us to calculate the maximum value of the previous subproblem for different values of $N.$ Let $N_p = \sum_{i=1}^{I_p} n_{i,p}$ be the number of individuals allocated to program $p$. Then the optimization problem above may be written as 
\begin{equation}
\label{eq:multiple_program_opt}
\begin{split}    
    & \text{maximize}_{N_p} \quad \sum_p^P F_p(N_p)\\
    & \text{subject to } N_{p} \geq 0, \sum_p N_p = N.
\end{split}
\end{equation}

where $F_p(N_p)$ is the solution of the A/B Testing Problem for program $p$, with prior $G_p$ and maximum allocation $N_p$. We may compute every value of $F_p(N_p)$ using the dynamic programming solution to the base A/B testing problem. Letting $\cl{M}(a,b)$ denote the solution to the problem above with $a$ programs and $b$ total units, dynamic programming can be used to solve \eqref{eq:multiple_program_opt}, with the update
\[
\cl{M}(a,b) = \max_{j=0}^b \left(\cl{M}(a-1,b-j) + F_a(j) \right)
\]
and base case
\begin{align*}
\cl{M}(1,b) & = F_1(b) \\
\cl{M}(a,1) & = \max\left(\cl{M}(a-1,1), F_a(1) + \sum_{p=1}^{a-1} F_p(0)\right).
\end{align*}

Figure \ref{fig:meta_production_byprogram} shows the metaproduction functions for the same three different Netflix experimentation programs in Figure \ref{fig:prior_densities}. The number of ideas for each program roughly reflects the relative sizes of each program in past data. The shape of the metaproduction functions immediately shows that we should not expect the optimal allocation to place all mass on one program. Solving the resulting dynamic programming gave an allocation of $[112, \  89, \  49]$ million units to the three experimentation programs respectively.

\paragraph*{Shared R\&D Resources.} A similar approach can answer questions on managing experimentation programs which share resources for generating ideas. Suppose we have $P$ different experimentation programs, which do not share allocation pools; within each program, tests have disjoint allocations. Let each program $p$ have access to an allocation pool of size $N_p$. Suppose further that the cost to create ideas is the same across programs, and there is budget to create $I$ additional ideas across all $P$ programs. 

The problem is to choose the number of ideas $I_p$ for each program $p$, the allocation $(n_{i,p})_{i=1,\dots,I_p}$ to test ideas for program $p$, and the subset $S_p$ of ideas to ship, in order to maximize cumulative return from all experimentation programs. This yields the optimization problem:

\begin{align*}
    & \text{maximize} \quad \E\left[\sum_{p=1}^{P} \sum_{i \in S_p} \Delta_i \right]\\
    & \text{subject to } \sum_{p=1}^P I_p \leq I.
\end{align*}

The solution of this problem shares the dynamic programming structure of the heterogeneous A/B testing problem. Let $F_p(I,N_p)$ be the optimal value of the base A/B testing problem with $I$ ideas to test and $N_p$ units. Then by similar arguments to the previous case, the optimization problem may be written as 
\begin{align}
\label{eq:resource_optimization}
\begin{split}
    & \text{maximize}_{I_p} \ \sum_{p=1}^P F_p(I_p, N_p)\\
    & \text{subject to } I_p \geq 0, \sum_{p=1}^P I_p \leq I.
\end{split}
\end{align}
It is also possible to weight each program differently in the objective, which may be useful if the programs target different metrics. Again, dynamic programming may be used to solve this. Fixing the total number of units $N$, let $\cl{M}(I,P)$ be the maximal value of the program \eqref{eq:resource_optimization}. By casework on the value of $I_p,$ the maximal value satisfies the equality
\[
\cl{M}(I,P) = \max_{j=0}^I \left( \cl{M}(I-j, P-1) + F_p(I, N_p)\right).
\]
The optimal allocation can be rediscovered by keeping track of the value $j$ which maximizes the quantity above. Figure \ref{fig:idea_meta_production_byprogram} shows the solution of the resource allocation problem for the same three programs as Figure \ref{fig:prior_densities}. Amongst three programs, the allocation is $[30,20,0]$. Although the solution suggests zero allocation should be given to Program $3$, the results would change if the allocation pool for the program were much larger, or if the prior were to change.



\begin{figure}
    \centering
    \begin{subfigure}{0.49\linewidth}
        \centering
        \includegraphics[width=\linewidth]{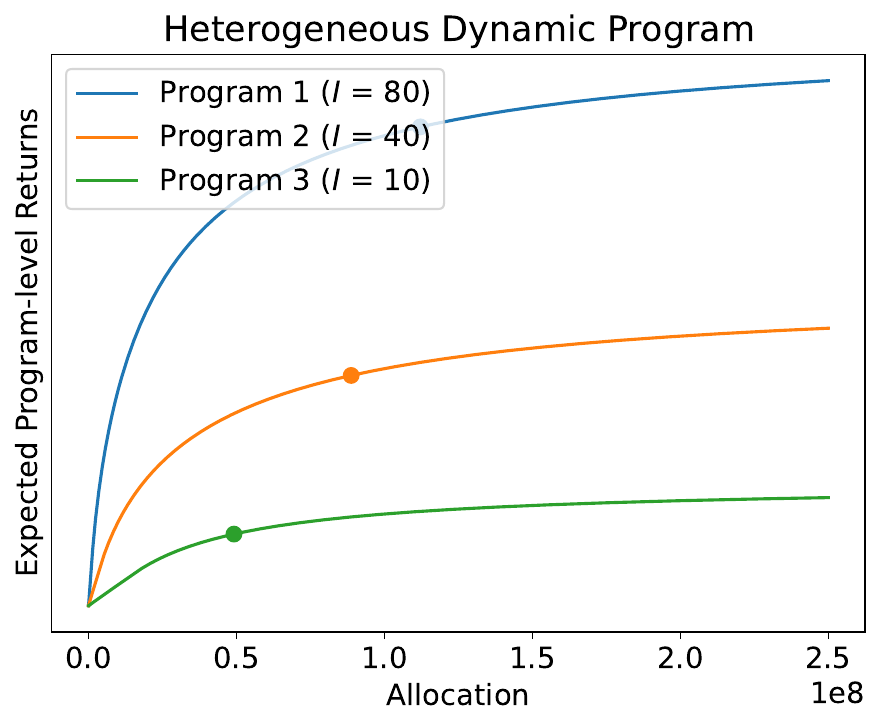}
        \caption{}
        \label{fig:meta_production_byprogram}
    \end{subfigure}
    \hfill
    \begin{subfigure}{0.49\linewidth}
        \centering
        \includegraphics[width=\linewidth]{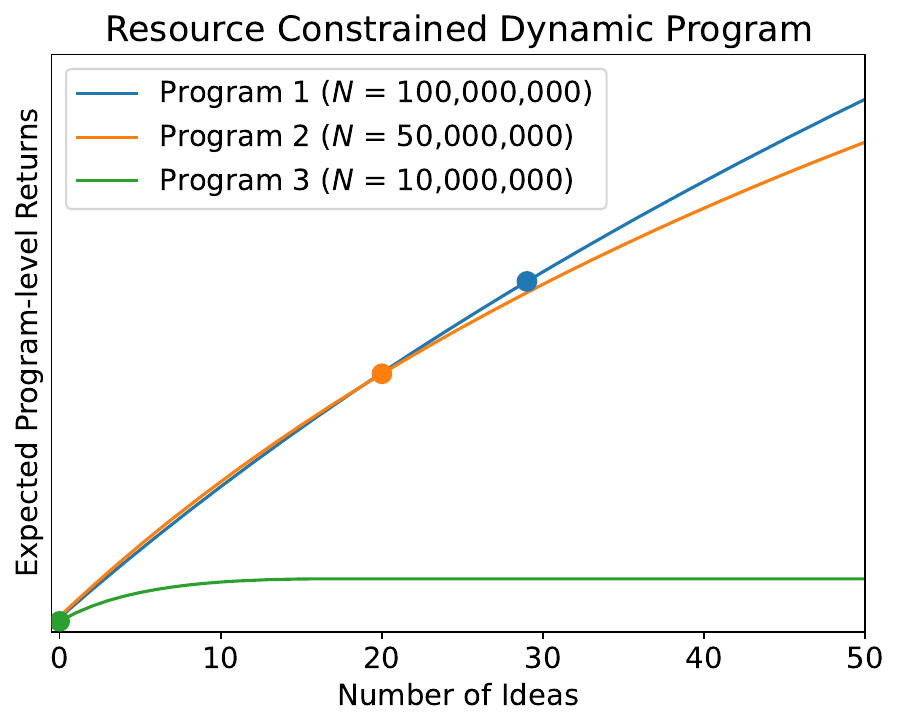}
        \caption{}
        \label{fig:idea_meta_production_byprogram}
    \end{subfigure}
    \caption{Left: Expected program-level return from three different experimentation programs. $y$-scale is obfuscated, and dots signify the optimal allocation to each program. The slopes of the curves at these points are all equal. Curves will change as the number of ideas change. Right: Additional return from three different experimentation programs, as a function of the number of ideas available for testing. Dots signify the optimal idea allocation to each program. Curves will change as the number of available units for testing changes. The allocation pools for each program are fixed and indicated in the legend.}
\end{figure}

\section{Beyond the Base Framework}

The goal of optimizing returns can be formalized by models other than the A/B Testing Problem in Section \ref{sec:setup}. For example, experimentation programs may enroll units in multiple tests at once, or may select at most one idea to be implemented, if the treatments are mutually incompatible. In another direction, a new experimentation program may not have enough data to estimate a prior $G$, raising the need for alternative decision procedures. In this section, we explore the resulting optimization problems and decision theory for these settings. 

\subsection{Constraints in A/B Testing Problems}
\label{sec:real_constraints}

Returning to the constraint $\sum_i n_i = N$, a more realistic constraint in many experimentation programs is to allow for units to be enrolled in multiple tests at once. Assuming that units are exchangeable for every test in an experimentation program and every unit can be enrolled at most $k$ times, one can solve the A/B Testing Problem by imposing the constraints $\sum_i n_i = kN, n_i \leq N$. Whenever these constraints are satisfied, there is an allocation of units to tests such that every unit is enrolled in at most $k$ tests,  each test has $k$ units, and units cannot be enrolled multiple times in the same test. The dynamic program can be solved in a similar fashion using now the update
\[
\cl{M}(i,n) = \max_{j=0}^{\min(n,N)} \left( \cl{M}(i-1, n-j) + f(j) \right)
\]
and evaluating $\cl{M}(I,kN).$

In practice, multiple experimentation programs share some allocation pools but not others, suggesting a different set of constraints in the problem. Suppose we have the following data:
\begin{enumerate}
    \item $k = 1,\dots,K$ unit allocation pools, of size $N_k$ each.
    \item $p = 1,\dots,P$ experimentation programs, with $i=1,\dots,I_p$ ideas each; the allocation for idea $i$ in program $p$ is given by $n^{(p)}_i$. Let $n_{i,k}^{(p)}$ be the number of units assigned to test idea $i$ in program $p$, from unit pool $k$.
    \item For each program $p,$ there is a set of indices $E_p \subseteq [K]$ keeping track of excluded unit pools.
    \item Assume in unit pool $k$, units can be enrolled in at most $c_k$ tests.
\end{enumerate}
Let $f_p(n)$ be the production function for the $p$th program. Maximizing returns subject to the constraints above can be formulated as the optimization problem

\begin{align}
\label{eq:general_opt}
\begin{split}
     \text{maximize} & \sum_{p=1}^P \sum_{i=1}^{I_p} f_p(n_i^{(p)})\\
     \text{subject to } & \sum_{p=1}^P \sum_{i=1}^{I_p} n_{i,k}^{(p)} \leq c_kN_k \\
     & n_i^{(p)} = \sum_{k=1}^K n_{i,k}^{(p)} \\
    & n_i^{(p)} \geq 0, 0 \leq n_{i,k}^{(p)} \leq N_k\\
    & n_{i,k}^{(p)} = 0, \forall i, \forall k \in E_p.\\
\end{split}
\end{align}

Under a Gaussian prior, let $\hat{x}_p$ be the threshold above which $f_p$ is concave. If we replace the constraint $n_i^{(p)} \geq 0$ by $n_i^{(p)} \geq \hat{x}_p$, the optimization problem \eqref{eq:general_opt} becomes a convex optimization problem. If $\hat{x}_p$ is small enough, the convex optimization proxy is a good approximation to problem \eqref{eq:general_opt}; one can round down $n_i^{(p)}$ to zero to avoid running tests with a very small number of units, for operational reasons.

\begin{remark}
The optimization problem is easy to implement in \textsf{cvxpy}, even with many constraints. There are also various ways to extend this insight.
\begin{enumerate}
    \item The same optimization problem can be used for a single experimentation program where ideas in the same program have different allocation pools.
    \item The optimization problem can be done whenever $f_p$ is concave; the results of \cite{azevedo2020b} suggests that when the prior is $t$-distributed with degrees of freedom less than $3$, the production function is concave. It also suggests we may optimize an objective consisting of a generic utility function $u$ applied to the sum of production functions, so long as $u$ is concave.
\end{enumerate}
\end{remark}

\subsection{Mutually Exclusive Treatments}

\label{sec:mutually_exclusive_treatments}

Consider a mutually exclusive version of the A/B testing problem, where only one test may be shipped amongst the pool of tested ideas.
\cite{azevedo2020b} study this problem as Extension 2 of their framework, in the Supplemental Information G.2 section of their paper. The optimal test $s$ to ship is given by

\[
\argmax_{i=1,\dots,I} \; \E[\Delta_i | \widehat{\Delta}_i]
\]

if the largest posterior expectation is positive, which is recorded in SI G.2 of \cite{azevedo2020b}. The mutually exclusive A/B testing problem with this objective reduces to the optimization problem.
\begin{align}
\label{eq:one_treatment_optimization}
\begin{split}
    \text{maximize} & \quad \E\left[\max_{i=1,\dots,I} \E[\Delta_i | \widehat{\mathbf{\Delta}}; n_i]^+ \right] \\
    \text{subject to } & \sum_i n_i = N.
\end{split}
\end{align}
By considering symmetric solutions which equally allocate across some number of tests, it suffices to choose $I_0 \leq I$ which maximizes
\begin{equation}
\label{eq:mut_production}
\E\left[\max_{i=1,\dots,I_0} \E[\Delta_i | \widehat{\mathbf{\Delta}}; N/I_0]^+ \right].
\end{equation}

Figure \ref{fig:mutually_exclusive_results} shows a plot of \eqref{eq:mut_production} as a function of $I_0$ for various values of $N$. In this case, the optimal number of tests to run is much smaller than in Section \ref{sec:allocation_practice}.

The expression is difficult to compute analytically, but we may provide useful approximations. Notice that $\max_{i=1,\dots,I_0} \E[\Delta_i | \widehat{\mathbf{\Delta}}; N/I_0]$ is equal to
\[
\frac{\frac{\sigma^2 I_0}{N}}{\tau^2 + \frac{\sigma^2 I_0}{N}}\mu + \frac{\tau^2}{\tau^2 + \frac{\sigma^2 I_0}{N}} \max_{i=1,\dots,I_0} \widehat{\Delta}_i.
\]
\sloppy By well-known arguments \cite{chatterjee2014superconcentration} for the concentration of maxima of Gaussian random variables, this suggests that $\max_{i=1,\dots,I_0} \E[\Delta_i | \widehat{\mathbf{\Delta}}; N/I_0]$ is exponentially concentrated around 
\begin{equation}
\label{eq:mutually_exclusive_mean}
\mu + \sqrt{\frac{2\log I_0}{\tau^2 + \frac{\sigma^2 I_0}{N}}}    
\end{equation}
with fluctuations of order $1/\sqrt{2\left(\tau^2 + \frac{\sigma^2 I_0}{N}\right)\log I_0}$. Whenever the mean \eqref{eq:mutually_exclusive_mean} is larger than several multiples of the fluctuation, we should expect that Eq. \eqref{eq:mut_production} is approximately equal to
\[
\E\left[\max_{i=1,\dots,I_0} \E[\Delta_i | \widehat{\mathbf{\Delta}}; N/I_0] \right] = \mu + \sqrt{\frac{2\log I_0}{\tau^2 + \frac{\sigma^2 I_0}{N}}}.
\]
The approximation \eqref{eq:mutually_exclusive_mean} can be used to inform the optimum value of $I_0 \leq I$ which maximizes expected return. Optimizing this in $I_0$, we obtain an approximation to the solution of the mutually exclusive A/B testing problem.


\begin{figure}
    \centering
    \includegraphics[width = 0.6\linewidth]{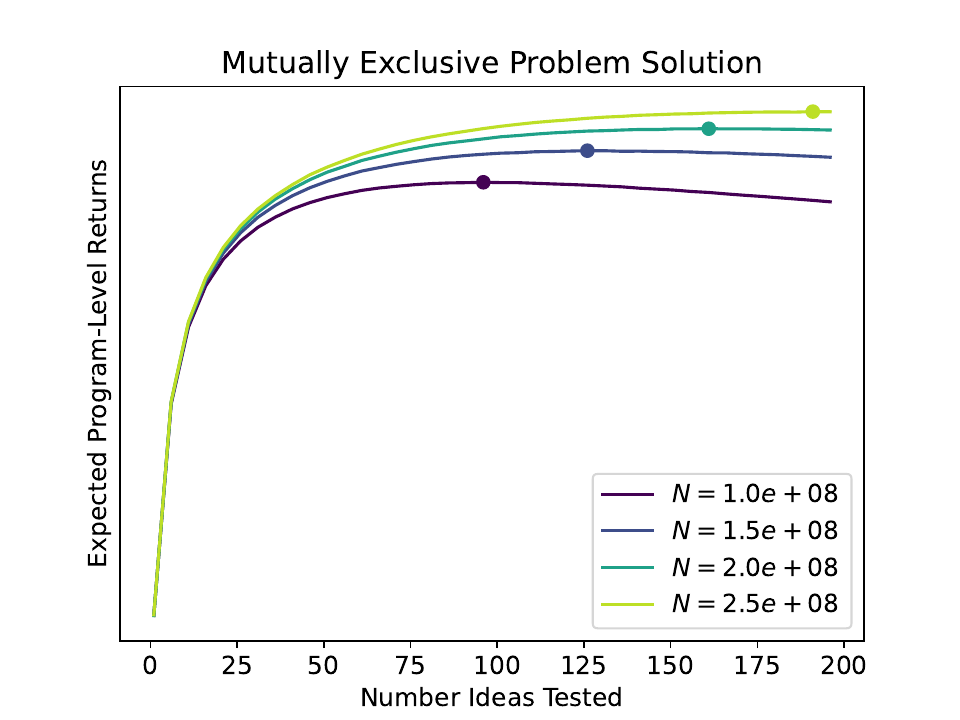}
    \caption{Expected metric value as a function of number of ideas tested. Curves denote different sizes of the allocation pool, and are generated with Monte Carlo. The optimal number of ideas to test varies significantly, and the conclusion may be different than lean experimentation. The concave shape of the curves may be explained intuitively as follows. Testing too little misses out on many very good ideas, while testing too much runs into winner's curse, as the test results become more noisy. } 
    \label{fig:mutually_exclusive_results}
\end{figure}

\subsection{Decision Problems for New Experimentation Programs}
\label{sec:new_exp_programs}

For newly launched experimentation programs, there is little data to estimate a prior $G.$ To navigate around issues of choosing a subjective prior, the goal of maximizing welfare $\sum_{i \in S} \Delta_i$ can be recast as finding the minimax decision for the loss
\[
\sum_{i} \Delta_i^+ - \sum_{i \in S} \Delta_i.
\]
Suppose for each $i$ we observe $X_i \sim \sf{N}(\Delta_i, \sigma^2/n_i)$ with fixed $n_i.$ The minimax decision, perhaps surprisingly, is to ship all $X_i > 0,$ which is a much less conservative decision than $p < 0.05$. The resulting decision is similar to the conclusions found in Section \ref{sec:implications}.

\begin{proposition}
\label{prop:minimax_optimal_multiple_tests}
The decision rule of shipping all $i$ for which $X_i \geq 0$ is minimax optimal.
\end{proposition}
In the single test case, this result is well-known in the statistical treatment literature in economics; see \cite{manski2004statistical, stoye2009minimax}, and further references in \cite{manski2019treatment}, which focuses on minimax regret problems. The result extends simply to multiple observations, with different variances, for which we provide a simple proof. See also \cite{chiong2023minimax, joo2023getting} for a more refined discussion of minimax decisions in A/B testing.

For a single test, Section 4 of \cite{azevedo2023b} derives the same expression for minimax risk. The expression is used to size a single experiment given external costs to enrolling members. We show how the same expression suggests a strategy for optimizing returns from a portfolio of A/B tests, suggesting the lean experimentation is optimal from a minimax perspective.

\begin{lemma}
\label{lemma:minimax_allocation}
Without costs to assigning members to experiments, the allocation which reduces the minimax risk is to allocate equally to all ideas.
\end{lemma}

An empirical Bayesian approach is also natural here, as in \cite{guo2020empirical}. If we are willing to assume $\Delta_i \sim G$ for an unknown prior which cannot be estimated from past data, the problem naturally becomes an empirical Bayes problem. Empirical Bayes for welfare maximization does not seem to have been studied explicitly, although some results on the nonparametric maximum likelihood procedure (NPMLE) can be applied; see Equation 3.7 of \cite{chen2022empirical}. It is unclear whether there is a matching regret lower bound and whether the plug-in NPMLE approach is optimal. We raise this as an interesting direction for future research.


\section{Discussion}
\label{sec:discussion}

Given a prior on treatment effects and the mandate to optimize returns, organizations conducting experimentation are typically better served by a return-aware framework rather than null hypothesis statistical testing. As a consequence -- with no costs --  organizations should be testing many more ideas, lowering the barriers to testing as much as possible, and reconsidering decision thresholds for shipping tests.

There are several reasons why the $p < 0.05$ paradigm remains entrenched in digital platforms. An appropriate combination of inertia, loss aversion, costs, and prior distribution could make conservative thresholds, in fact, optimal. However, running the exercises outlined in Section \ref{sec:costs} for Netflix experimentation programs, we find that individually varying each of these factors suggests either an unrealistic prior, an unusual degree of risk-aversion, or overly large implied costs. Although the results will change per firm, in our experience, many firms do not justify current practice with informed cost measurements or decision theoretic considerations. They rather default to statistical tradition.

One widespread argument for conservative decision making is the cost of false positives in A/B tests \cite{berman2022false, kohavi2024false}.    As part of relaxing $p$-value thresholds, the return-aware framework necessarily eschews false positive rate control. Figure \ref{fig:test_passing} reflects this. It is important to note that most of these false positives have treatment effects of small magnitude. 

The primary concern around false positives, even if the magnitude is small, centers around the risk of building future ideas on false positive discoveries, which may result in a sub-optimal path of innovation relative to a counterfactual based on a true discovery. But this appears to confuse decision making and epistemology: one can imagine a simple framework where we make decisions using a return-optimization framework, but gain knowledge about the platform and population using results which are statistically significant. That is, we only invest in ideas which pass a higher bar of statistical evidence. Overall, there appears to be a lack of empirical or theoretical study for the impact of false positives on future innovation; we raise this as an interesting direction for future research. 

Of course, the return-aware framework should not always be applied. For experiments in practice which are designed to learn about the population of interest, a null hypothesis testing framework with a conservative decision threshold applies. For atypical experiments that are not plausibly drawn from $G$, one can default to a minimax framework as in Section \ref{sec:new_exp_programs} or the traditional hypothesis testing framework. 

We also emphasize the importance of choosing the right metric to optimize in an experimentation program. Optimizing short-term metrics in the online technology industry runs into the risk of creating negative externalities for users, and potential backlash. In this way, the framework pairs well with recent research on surrogates for long-term outcomes.

Finally, we note that the return-aware framework is potentially useful in other settings beyond A/B Testing. For example, experimentation teams in political campaigns desire to maximize donations or voter turnout. Economic field experiments can be used in service of maximizing social welfare. Agricultural experiments aim to maximize crop yields. In the first two settings, there is an important additional nuance in the literature, which focuses on the generalization of the sample ATE to the population ATE, as well as a focus on policy learning. See for example \cite{manski2004statistical, kitagawa2018should, athey2021policy, ben2022policy} and the references therein; some of the works focus on the observational case, which is different from our setting. See also \cite{milkman2021megastudies} for an experimentation program in behavioral science with multiple parallel experiments, where our setting may apply.

In the remainder of this section, we conclude with a list of directions for future research.

\subsection{Directions for Future Research}
\label{sec:open_problems}

\paragraph*{Estimation of $G$.} We focus primarily on the problem of optimizing future test design, having knowledge of $G$. Under a fully Bayesian specification, multi-level hierarchical models such as \cite{ejdemyr2023estimating} can be used to obtain the posterior on $G$ given past data on A/B tests. From a frequentist perspective, $G$ needs to be estimated from past data; we opt for this approach, using maximum likelihood estimation for a simple normal parametric family. \cite{azevedo2020b} uses maximum likelihood for a $t$-distributed parametric family with shape and location parameters; see \cite{azevedo2019empirical} for an overview of methods to estimate $G$.

There are several interesting statistical questions in the estimation of $G$, since the quantity of interest we primarily care about is the production function $f(n).$ Firstly, one should understand the risk of incorrectly estimating $G$ from past data, under a parametric or nonparametric model. The estimated tail behavior is important, based on the reasoning in \cite{azevedo2020b}. Secondly, understanding the value of non-parametric approaches, such as the NPMLE, is important. Although the NPMLE learns priors which are atomic, it would be interesting to investigate to what extent the NPMLE can be used to learn the production function $f(n)$. This avenue is promising, as the NPMLE (e.g. \cite{koenker2014convex}) has several adaptivity properties that make it suitable to learning both light-tailed and heavy-tailed priors \cite{shen2022empirical, soloff2024multivariate, jiang2020general}, although results on estimation rates for $G$ are still currently being researched \cite{soloff2024multivariate}. Notice by Tweedie's formula that the production function may be written as 
\begin{equation}
\label{eq:prod_fcn_f}
f(n) = \int \left(xm_{\hat{\Delta},\frac{\sigma^2}{n}}(x) +  \frac{\sigma^2}{n} m'_{\hat{\Delta},\frac{\sigma^2}{n}}(x) \right)^+  dx,
\end{equation}
where $m_{\hat{\Delta},\sigma^2/n}(x)$ is the marginal density of an observation $\hat{\Delta}_i$ with variance $\sigma^2/n.$ Eq. \eqref{eq:prod_fcn_f} appears to be useful in obtaining error bounds when estimating $G$ by NPMLE, using the results of \cite{jiang2009general, jiang2020general, soloff2024multivariate}, or by other estimates \cite{guo2020empirical}. Deconvolution kernel estimators \cite{fan1991optimal, delaigle2008density, meister2009density} may also be used here, although they require tuning parameters for the heteroskedastic problems in our setting. NPMLE, on the other hand, is tuning-free.

In estimating $G$, we ignore the multi-arm structure of most experiments in practice. This leaves a lot of data on the table. Estimation of the prior $G$ with multi-arm structures can be done with a two-level hierarchical model; techniques such as the NPMLE can be used here, or fully Bayesian models such as \cite{ejdemyr2023estimating} may be used. After estimating $G$, there is the natural question of how to optimize returns from a portfolio of multi-cell experiments. 


\paragraph*{Sequential Frameworks.} The A/B Testing problem is a one-shot optimization. In practice, we have multiple time periods to space out tests. Incorporating time into the problem suggests several tradeoffs. For example, there would be a tradeoff between running every idea with high powered experiments, taking a long time, versus testing all ideas at once with very under-powered tests. Here is one model which captures some of aspects of these tradeoffs. Suppose that we have $t = 1,\dots,T$ time periods, $I$ total ideas to test, and $N$ units. During a time period $t$, we may select a number of ideas to test using $N$ units; for a given period, let $I_t$ be the number of ideas tested. The problem is to choose an allocation $(n_{i,t})_{i=1,\dots,I_t}$ to test ideas at time $t$, the idea allocations $I_t$, and the subset $S_t$ of ideas to ship, in order to maximize cumulative return by a fixed time, for any choice of weights $w_t$:

\begin{align*}
    & \text{maximize} \quad \E\left[\sum_{t=1}^{T} w_t \sum_{i \in S_t} \Delta_i \right]\\
    & \text{subject to } \sum_{t=1}^T I_t \leq I.
\end{align*}

The problem may be solved using dynamic programming in a fashion similar to the problems in Section \ref{sec:multiple_programs}. Figure \ref{fig:sequential_dp_solution} shows the solution of the sequential A/B testing problem. Such insights are useful for planning and scheduling exercises around experimentation. For a metric which upweights ideas that are shipped earlier, it is intuitive that more tests are shipped early on. As $T$ increases, one should expect the allocation of tests over time to even out.

\begin{figure}
    \centering
    \includegraphics[width = \linewidth]{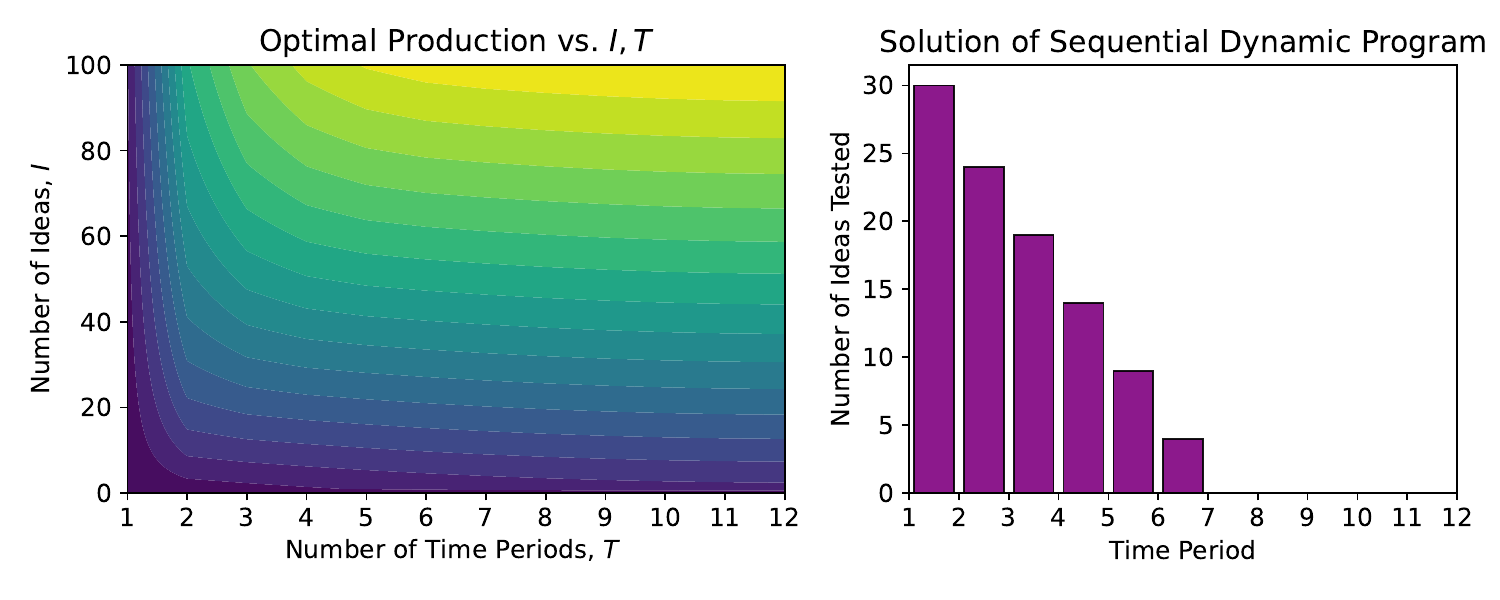}
    \caption{Solution of a sequential dynamic program for a Netflix experimentation program.  Lighter colors indicate larger returns. The metric $\sum_t (T- t) \sum_{i \in S_t} \Delta_i$ was used, to account for additional gains when a treatment is shipped sooner. Top: optimal expected return when testing $I$ ideas over $T$ time periods. Bottom: optimal testing schedule for $I = 100, T = 12.$}
    \label{fig:sequential_dp_solution}
\end{figure}

The sequential version of the A/B testing problem we introduce allows managers to answer scheduling questions. However, it misses an important aspect of A/B test decision making - test lengths are variable, and there is an option of waiting longer to collect more information on a treatment. Sequential decision making for a single decision has been studied by Wald and collaborators \cite{wald1992sequential, wald1950bayes, arrow1949bayes}, and has seen intense resurgence in the context of digital platforms \cite{maharaj2023anytime, ham2022design, johari2022always}. The problem with multi-arm tests is closely related to the multi-armed bandits, best-arm identification, and sequential clinical trials literatures, depending on the loss function specified. See \cite{kasy2021adaptive, adusumilli2021risk, adusumilli2022neyman, adusumilli2022sample} and the literature reviews within. From the perspective of experimentation programs, there is the interesting extension of how to manage a portfolio of A/B tests sequentially. Under the assumption of no cross-allocation, at each time one must choose an allocation of units, and for each test, there is a decision to either ship the idea or continue testing. Deriving optimal strategies to maximize cumulative returns would be of clear interest.

\paragraph*{Data-driven Program Groupings.} Operationally, it is easy to define collections of A/B tests which should naturally be grouped together. But this process is easily criticized - are there other groupings of A/B tests for which the assumption of i.i.d. draws is more realistic?  This is the \textit{relevance problem} of empirical Bayes \cite{efron2012large}, which is also highlighted in the conclusion of \cite{deng2015objective}. Experimental metadata could be used as covariates to inform a ``data-driven" grouping. See \cite{ignatiadis2019covariate} for a related approach.

\paragraph*{Modelling Idea Generation.} In the setup considered so far, the firm has $I$ ideas to test and the generative process for arriving at these ideas is left unmodelled. As hinted earlier, colloquial discussions about the prevalence of conservative launch rules are motivated by seeking to avoid false-discovery ``trajectories". Overall, the cost of false positives on the rate of future innovation appears to be understudied. Is there empirical evidence or simple theoretical models to support this phenomenon? Simple dynamic models of innovation such as \cite{callander2011searching} could be used to study extensions of the A/B testing problem with this issue in mind.

Furthermore, the current work abstracts from the distribution of ideal \emph{quality} when testing an order-of-magnitude larger number of ideas; it may plausibly be the case that the average quality of each additional idea may diminish in certain experimentation programs, which may provide important additional considerations in practice.

\paragraph*{Diminishing Returns and Data Relevance.} Industry wisdom suggests that for an optimized product, most ideas should have small impact. For relatively unoptimized products, many ideas result in big wins. Thus, older data should be less relevant for estimating a prior; thus points the need for research on how best to weight past data under slow distribution shift.

Moreover, if the distribution $G$ is re-estimated in the future, there is a concern that allocation which optimizes returns to experimentation is at odds with the allocation that best allows the manager to estimate $G.$ In the setting with a normal prior, what is the optimal allocation to best estimate $\mu,\sigma^2$? Are a few high-powered experiments better than many low-powered experiments for estimating the prior? Consider allocations which equally split $N$ units across some fixed number of tests $I_0$, and zero on the others. Then the estimates $\widehat{\Delta}_i$ are i.i.d from $\textsf{N}\left(\mu,\tau^2 + \frac{\sigma^2 I_0}{N}\right)$; the MLE estimates for $\mu,\tau^2$ reduce to the sample mean and sample variance, so that
\begin{align}
    \Var(\hat{\mu}_{\sf{MLE}}) & =  \frac{1}{I_0}\left( \tau^2 + \frac{\sigma^2 I_0}{N} \right)  \\
    \Var(\hat{\tau}^2_{\sf{MLE}}) & = \frac{1}{I_0 - 1}\left( \tau^2 + \frac{\sigma^2 I_0}{N} \right)^2.
\end{align}
Thus $\Var(\hat{\mu}_{\sf{MLE}})$ is strictly decreasing in $I_0$, and $\Var(\hat{\tau}^2_{\sf{MLE}})$ is minimized when $I_0 \approx \frac{\tau^2N}{\sigma^2}$. The optimal $I_0$ values for either the standard deviation and variance were typically much larger than the number of tests currently being run at Netflix. As a result, increasing the number of experiments being run optimizes both returns to experimentation and the ability to learn the prior $G$. This computation suggests similar tradeoffs for more interesting estimators, which could be a direction for future research.



\section{Acknowledgements}

We thank Eduardo Azevedo, Travis Brooks, Michael Howes, Michael Lindon, Tomoya Sasaki, and attendees of CODE@MIT 2024 for helpful comments and discussion.

\begin{appendix}

\section{Related Work on the Analysis of A/B tests and Decision Theory} 
\label{sec:litreview}

There have been several papers on Bayesian / empirical Bayes approaches to managing A/B tests \cite{goldberg2017decision, guo2020empirical, deng2015objective, abadie2023estimating, ejdemyr2023estimating}. The models of test returns closely mimic those of \cite{azevedo2020b, azevedo2023b}, with different objectives and methods. \cite{deng2015objective} compares Bayesian and frequentist approaches to A/B testing and proposes empirical Bayes as a middle ground that borrows strength from both approaches. Expectation maximization is used to estimate a prior from historical data on A/B tests. However, \cite{deng2015objective} focuses on pure hypothesis testing with an empirical Bayes prior. 

Towards the goal of optimizing returns, \cite{feit2019test} give methods for designing A/B tests under a two stage procedure, solving a similar problem of how to maximize profit in an A/B test by optimizing the sample size. This paper moves towards a dynamic approach to the problem, and uses a normal-normal model under an allocation constraint. \cite{goldberg2017decision} gives a decision theoretic approach similar to that advocated in \cite{manski2019treatment} for A/B testing. The authors give a purely Bayesian approach, using a risk function that weights the magnitude of treatment effects. They use hierarchical Bayesian models to fit the posterior given data on past A/B tests at eBay. 

\cite{guo2020empirical} focuses on the problem of estimating the prior $G$ using an interesting approach blends the $F$-modelling and $G$-modelling of empirical Bayes \cite{efron2012large}. For data on past A/B tests with heteroskedastic variances, they show that a heat equation characterizes the marginal density of the data, with variance as a parameter. A spectral approach for fitting maximum likelihood estimate of $G$ via trigonometric polynomials is proposed. The solution can be found efficiently by convex optimization, and consistency is proven. \cite{coey2019improving} gives an alternative approach to modelling $G$ by \textit{experiment splitting}. The population for an A/B test is split into two groups, and the mean of one group is nonparametrically regressed on the other. \cite{coey2019improving} shows that the regression estimates a quantity similar to the posterior expectation. \cite{ignatiadis2023empirical} gives a related approach generalizing \cite{coey2019improving} in empirical Bayes problems with multiple replicates.

The problem of selecting all treatments to ship is somewhat related to problems of choosing the best $k$ treatments and rankings. See \cite{coey2022empirical} for an empirical Bayes approach to a more general problem of picking the top $k$ treatments. \cite{gu2023invidious} solves ranking problems with empirical Bayes methods.

The perspective of optimizing returns to experimentation follows naturally from a line of work which first \textit{estimates} the returns to experimentation \cite{abadie2023estimating, ejdemyr2023estimating}. \cite{abadie2023estimating} introduces an empirical Bayes procedure to measure the returns to evidence-based decision making.\cite{ejdemyr2023estimating} use a hierarchical Bayesian model to provide estimates of cumulative returns to experimentation which correct for winner's curse. 

The goal of optimization and the associated decision theory is related to approaches for going beyond $p < 0.05$ tradition; see \cite{wasserstein2019moving} for a broad discussion comprising Vol. 73 of the \textit{American Statistician}. As part of this volume, \cite{ruberg2019inference} provides a Bayesian approach similar to ours in the context of clinical trials, which is aware of treatment effect magnitude. The volume also includes the previously-cited reference \cite{manski2019treatment}, which concisely advocates the need for decision theory in problems where we care about treatment choice. \cite{manski2019treatment} also connects the problem to the economic literature on statistical treatment rules, which offers frequentist approaches to similar problems of decision making for maximizing welfare. See e.g. \cite{manski2004statistical, stoye2009minimax, stoye2012minimax, hirano2009asymptotics, tetenov2012statistical}.

\bibliographystyle{alpha}
\bibliography{bibliography} 

\section{Proofs}



\subsection{Section \ref{sec:pvalue_framework} Proofs}

\begin{proof}[Proof of Proposition \ref{prop:pvalue_framework_equivalence}]
Lemma \ref{lemma:monotonicity} shows that the conditional expectation $\E[u(\Delta_i) | \hat \Delta_i]$ is monotone increasing in $\hat\Delta_i$. As a result, there exists a constant $c$ such that 
\[
\hat{\Delta}_i \geq c \Leftrightarrow \E[u(\Delta_i) | \hat \Delta_i] \geq 0.
\]
Because $1 - \Phi(x)$ is decreasing, there exists a threshold $\alpha$ such that 
\[
1 - \Phi\left(\frac{\widehat{\Delta}_i\sqrt{n}}{\sigma} \right) \leq \alpha \Leftrightarrow \E[u(\Delta_i) | \hat \Delta_i] \geq 0
\]
as desired.

\end{proof}

\subsection{Section \ref{sec:metaproduction} Proofs}

\begin{proof}[Proof of Theorem \ref{thm:metaprod_description}]
Reparametrize the optimization problem by the variable $x = N/i$. This yields
\begin{equation}
    F(I,N) = N \max_{x \in \left[\frac{N}{I},N \right]} \frac{f(x)}{x}.
\end{equation}
Next, we claim that the function $g(x) := f(x)/x$ on $\mathbb{R}^+$ is unimodal: it is positive everywhere, increasing on $[0,x^*]$ and decreasing to zero on $[x^*,\infty].$ If this is true, it is clear that Eq. \eqref{eq:metaprod_description} holds. \\

Positivity is clear. To prove the other claim, first note $g$ is smooth and $g'(x) = \frac{xf'(x) - f(x)}{x^2}$. Thus it suffices to show there exists an $x^*$ such that $h(x) := xf'(x) - f(x) \geq 0$ on $[0,x^*]$ and is non-positive on $[x^*,\infty).$ The derivative of $h$ is equal to $xf''(x)$. Recall that Proposition 2 of \cite{azevedo2023b} shows there exists $\hat x$ for which $f(x)$ is convex on $[0,\hat x]$ and concave on $[\hat x, \infty).$ Thus, $h$ is increasing on $[0,\hat x]$ and is decreasing on $(\hat{x},\infty)$. Finally, note the following two properties of $h.$ 
\begin{enumerate}
    \item $h(0) \geq 0$. Interpret $h(0)$ as the limit of the derivative $xf'(x)$ as $x \rightarrow 0^+$. A straightforward calculation shows that this exists; see also Proposition 1 of \cite{azevedo2020b}. Further it is nonnegative as $f'(x) \geq 0$ for all $x \geq 0.$
    \item $\liminf_{x \rightarrow \infty} h(x) < 0$.  since $f(x)$ is increasing and bounded above, $g(x)$ is decreasing for large $x$. Thus, for $x$ large enough, $h(x)$ is negative.
\end{enumerate}
By the intermediate value theorem, there exists such a $x^*$, greater than $\hat{x}$.
\end{proof}

\begin{proof}[Proof of Corollary \ref{cor:metaproduction_marginals}]
For the first part, $F_1(I)$ is increasing by a coupling argument: any testing strategy for $I$ ideas is admissible for $I' > I$ ideas. Now consider the following two cases. If $N \leq x^*$, then the function is constant and equal to $f(N)$. Otherwise Theorem \ref{thm:metaprod_description} implies
\begin{equation}
F_1(I) = 
\begin{cases}
    I \; f\left(\frac{N}{I} \right) & \text{if } I \leq \frac{N}{x^*} \\
    N \frac{f(x^*)}{x^*} & \text{if } I \geq \frac{N}{x^*}.
\end{cases}
\end{equation}
For $I$ larger than $\frac{N}{x^*}$, the function is constant. Otherwise, the function $g(I,N) = If\left( \frac{N}{I}\right)$ is the perspective transform of $f$. Since $I \in [1,\frac{N}{x^*}]$, the argument $\frac{N}{I}$ lies in $[x^*,N].$ The proof of Theorem \ref{thm:metaprod_description} shows that $x^* \geq \hat{x}$, and so $f$ is concave on $[x^*,N]$. Since the perspective function of a concave function is concave, we conclude that $F_1(I)$ is concave. \\

For the second part, recall that $x^*$ is independent of $I,N$. It only depends on the underlying distribution $G.$ Re-expressing the result of Theorem \ref{thm:metaprod_description} gives
\begin{equation}
F_2(N) = 
\begin{cases}
    f(N) & \text{if } N \leq x^* \\
    N\frac{f(x^*)}{x^*} & \text{if } N \in [x^*,Ix^*] \\
    I\; f\left( \frac{N}{I}\right) & \text{if } N \geq Ix^*.
\end{cases}
\end{equation}
Using the fact that $\hat{x} \leq x^*$, the desired properties are now easy to see.
\end{proof}



\begin{proof}[Proof of Corollary \ref{cor:regret_with_oracle}]

$L_1$ is equal to $cI - F(I,N),$ where $c = \E\left[\Delta_i^+\right]$ is explicit. By Theorem \ref{thm:metaprod_description}, 
\[
L_1 =  
\begin{cases}
    cI - f(N) & \text{if } x^* \geq N \\
    cI - N\frac{f(x^*)}{x^*} & \text{if } x^* \in [\frac{N}{I},N] \\
    I \left( c -  f\left( \frac{N}{I}\right) \right) & \text{if } x^* \leq \frac{N}{I}.
\end{cases}
\]
Since $N \rightarrow \infty$ and $N/I \rightarrow \kappa$, only the second and third cases are relevant. If $\kappa \leq x^*$ then eventually $x^* \in [N/I,N]$, and so $L_1 = cI - N\frac{f(x^*)}{x^*}$. In this regime, $L_1/I \rightarrow c - \kappa \frac{f(x^*)}{x^*}.$ The other case is similar, using continuity of $f$.
\end{proof}

\begin{proof}[Proof of Corollary \ref{cor:metaproduction_vs_mu_tau}]
This leverages Prop. 4 of \cite{azevedo2023b}; the result states that $\frac{\partial f}{\partial \mu} > 0$ if and only if $\mu < 0$, and $\frac{\partial f}{\partial \tau} > 0.$ By Theorem \ref{thm:metaprod_description}, $F(I,N)$ for fixed $I,N$ is some positive multiple of $f$ with argument depending on $I,N$. Hence it is also increasing as $\mu \uparrow 0, \tau \rightarrow \infty.$ 
\end{proof}

\subsection{Section \ref{sec:discussion} Proofs}

\begin{lemma}
\label{lemma:minimax_bounded_below_bybayes}
The minimax risk is bounded below by the Bayes risk of any prior.
\end{lemma}
\begin{proof}
For any decision $\delta$, $\sup_{\theta'} R(\theta',\delta) \geq R(\theta,\delta)$. Taking the expectation over $\theta \in G$ for some prior $G$, we find
\[
\sup_{\theta'} R(\theta',\delta) \geq \E_{\theta}R(\theta,\delta).
\]
Thus $\sup_{\theta'} R(\theta',\delta) \geq \inf_{\delta'}\E_{\theta}R(\theta,\delta'),$ and now taking the inf over the left hand side over $\delta$, we conclude
\[
\inf_\delta \sup_{\theta'} R(\theta',\delta) \geq \inf_{\delta'}\E_{\theta}R(\theta,\delta'),
\]
which shows that the minimax risk is bounded below by the Bayes risk over any prior.
\end{proof}

\begin{proof}[Proof of Prop. \ref{prop:minimax_optimal_multiple_tests}]
Firstly, the maximum risk of the decision rule to ship when $X_i > 0$ is given by
\begin{align*}
& \max_{\mu_i} \sum_i  \mu_i^+ - \mu_i \Phi\left(\frac{\mu_i\sqrt{n_i}}{\sigma}\right) \\
= & \sum_i  \max_{\mu_i > 0} \mu_i \left( 1 - \Phi\left(\frac{\mu_i\sqrt{n_i}}{\sigma}\right) \right).
\end{align*}
Recall from Lemma \ref{lemma:minimax_bounded_below_bybayes} that the minimax risk is bounded below by the Bayes risk for any prior. Consider a prior where $\mu_i$ are all i.i.d. and take values $\pm C_i$ with probability $1/2$. By independence, the Bayes optimal rule is to apply the Bayes rule for a single test individually. The Bayes risk is then given by the sum of the Bayes risks for the single test decision problem. 

For a single test decision problem, the optimal decision rule is to ship if $X_i > 0$: writing out an explicit formula for the posterior mean shows $\E[\mu_i | X_i = 0] = 0$ and the posterior mean is increasing. The Bayes risk of the single procedure is then $\E[\mu_i^+] - \E[\mu_i \mathbf{1}\set{X > 0}]$, which can be computed as  
\[
\frac{C_i}{2}\left(1 - \Prob\left(\sf{N}(0,\sigma^2/n_i) \in [-C_i,C_i]\right)\right)
\]
This simplifies to $C_i\left(1 - \Phi(\frac{C_i\sqrt{n_i}}{\sigma})\right)$. Picking the parameters $C_i$ to maximize the Bayes risk, the resulting quantity is
\[
\sum_{i=1}^n \max_{C_i > 0} C_i \left( 1 - \Phi\left(\frac{C_i\sqrt{n_i}}{\sigma} \right) \right),
\]
 which exactly matches the risk of the posited decision rule. 
\end{proof}

\begin{proof}[Proof of Lemma \ref{lemma:minimax_allocation}]
The expression for the minimax risk is given in the proof of Prop. \ref{prop:minimax_optimal_multiple_tests}. We claim that the function
\[
g(n) = \max_{\mu > 0} \mu \left(1 - \Phi \left(\frac{\mu \sqrt{n}}{\sigma} \right) \right)
\]
is concave decreasing. The decreasing property is simple to see, since the function $f_n(\mu) = \mu \left(1 - \Phi \left(\frac{\mu \sqrt{n}}{\sigma} \right) \right)$ is uniformly smaller as $n$ increases. To see concavity, change variables in the maximization by using $\nu = \mu\sqrt{n}/\sigma$, so that 
\[
g(n) = \frac{\sigma}{\sqrt{n}} \max_{\nu > 0} \nu\left(1 - \Phi(\nu) \right) = C/\sqrt{n},
\]
for some constant $C > 0.$ Combined with the fact that $n_i \geq 1$ (otherwise the minimax risk would be infinite), the problem becomes a convex optimization problem symmetric in the variables, showing that equal allocation reduces the minimax risk.

\end{proof}

\subsection{Misc. Computations}

\begin{lemma}[Monotonicity of the Posterior Expectation]
\label{lemma:monotonicity}
Let $F$ be an increasing function. Suppose $\mu \sim G$, with $\int |F(x)| dG(x) <\infty,$ and $X \sim \sf{N}(\mu,\sigma^2)$. Then the conditional expectation $\E[F(\mu) | X = x]$ is increasing in $x$.
\end{lemma}

\begin{proof}[Proof of Lemma \ref{lemma:monotonicity}]
The conditional expectation can be written out explicitly as 
\[
\frac{\int F(\mu) \phi(\frac{\mu - x}{\sigma}) dG(\mu)}{\int \phi(\frac{\mu - x}{\sigma}) dG(\mu)}
\]
By the dominated convergence theorem, using the boundedness of the Gaussian density, the quantity is differentiable in $x$. The derivative is a positive constant times
\begin{align*}
& \int \phi\left(\frac{\mu-x}{\sigma}\right) dG(\mu) \int F(\mu)\left(\mu-x \right)\phi\left(\frac{\mu-x}{\sigma}\right) dG(\mu) \\
& - \int F(\mu) \phi\left(\frac{\mu-x}{\sigma}\right) dG(\mu) \int \left( \mu-x\right)\phi\left(\frac{\mu-x}{\sigma}\right) dG(\mu).
\end{align*}
Renormalizing by another positive constant, this is equal to 
\begin{align*}
& \E[F(\mu)(\mu-X)|X] - \E[\mu-X|X]\E[F(\mu) | X].
\end{align*}
Since $F(\mu)$ and $\mu - x$ are increasing in $\mu$, the Harris-FKG inequality shows that this expression is positive. Thus, the derivative is positive.
\end{proof}


\end{appendix}

\end{document}